\documentclass[11pt]{article}
\usepackage{amsmath,amsxtra,amssymb,latexsym, amscd,amsthm,pb-diagram,fancyhdr,euscript}
\usepackage{amsthm}
\usepackage{latexsym}
\usepackage{amssymb}
\usepackage{amsmath}
\usepackage{indentfirst}
\usepackage{hyperref}
\usepackage{mathrsfs}
\usepackage{array}
\usepackage{euscript}
\usepackage{slashed}
\usepackage[T1]{fontenc}
\usepackage{graphicx}
\usepackage[utf8]{inputenc}
\usepackage[french,english]{babel}
\usepackage{lipsum} 
\usepackage{multicol}
\hypersetup{
   colorlinks,
    citecolor=blue,
    filecolor=black,
    linkcolor=blue,
    urlcolor=magenta
}
\hypersetup{linktocpage}

\newcommand{\ba}{\mathbf{a}}
\newcommand{\bb}{\mathbf{b}}
\newcommand{\bc}{\mathbf{c}}
\newcommand{\bd}{\mathbf{d}}
\newcommand{\be}{\mathbf{e}}

\renewcommand{\d}{\mathrm{d}}

\newcommand{\hook}{\lrcorner}
\newcommand{\hnabla}{\hat{\nabla}}

\newcommand{\thorn}{\mbox{\th}}

\newcommand{\scri}{{\mathscr I}}

\newtheorem{defn}{Definition}[section]
\newtheorem{thm}{Theorem}[section]

\newtheorem{lem}{Lemma}[section]
\newtheorem{rem}{Remark}[section]
\begin{document}\mbox{}

\vspace{0.25in}

\begin{center}
{\huge{\bf Peeling of Dirac fields on Kerr spacetimes}}

\vspace{0.25in}

\large{{\bf PHAM Truong Xuan\footnote{Email: xuanpt@tlu.edu.vn or phamtruongxuan.k5@gmail.com}}}

{\it Faculty of Information Technology, Department of Mathematics, Thuyloi university,\\  
175 Tay Son, Dong Da, Ha Noi, Viet Nam.}
\end{center}

\begin{abstract}
In a recent paper with J.-P. Nicolas [J.-P. Nicolas and P.T. Xuan, Annales Henri Poincar\'e 20(10):3419-3470, 2019. arXiv:1801.08996], we studied the peeling for scalar fields on Kerr metrics. The present work extends these results to Dirac fields on the same geometrical background. We follow the approach initiated by L.J. Mason and J.-P. Nicolas [L. Mason and J.-P. Nicolas, J Inst Math Jussieu 8(1):179–208, 2009. arXiv:gr-qc/0701049; L. Mason and J.-P. Nicolas, J Geom Phys 62(4):867–889, 2012. arXiv:1101.4333] on the Schwarzschild spacetime and extended to Kerr metrics for scalar fields. The method combines the Penrose conformal compactification and geometric energy estimates in order to work out a definition of the peeling at all orders in terms of Sobolev regularity near $\scri$, instead of ${\cal C}^k$ regularity at $\scri$, then provides the optimal spaces of initial data such that the associated solution satisfies the peeling at a given order. The results confirm that the analogous decay and regularity assumptions on initial data in Minkowski and in Kerr produce the same regularity across null infinity. Our results are local near spacelike infinity and are valid for all values of the angular momentum of the spacetime, including for fast Kerr metrics. 
\end{abstract}



\tableofcontents

\section{Introduction}
The Peeling is a type of asymptotic behaviour of zero rest-mass fields initially discovered by R. Sachs  \cite{Sa61,Sa62}. Its initial formulation involved an expansion of the field in powers of $1/r$ along a null geodesic going out to infinity, and the alignment of a certain number of principal null directions of each term in the expansion along the null geodesic considered. R. Penrose introduced the conformal technique in the early 1960's \cite{Pe63,Pe64} and used it to establish that the peeling property is equivalent to the mere continuity of the rescaled field at null infinity \cite{Pe65}. He also suggested that the peeling property is generic on asymptotically flat backgrounds and provided the typical model of a spacetime on which the peeling should be valid. These are the so-called asymptotically simple spacetimes whose Weyl spinor satisfies the peeling property, which yields a smooth (or possibly ${\cal C}^k$) null infinity. The genericity of the peeling behaviour raised questions at the time and it was believed that the Schwarzschild metric, with its different asymptotic structure from Minkowski spacetime, would impose more stringent constraints on initial data for the corresponding solutions to peel. This question was resolved by L. Mason and J.-P. Nicolas in \cite{MaNi2009,MaNi2012} in which it is shown that the same regularity and decay assumptions on the initial data in Minkowski and Schwarzschild spacetimes yield the same regularity at null infinity. The peeling for linearized gravity and for full gravity has been studied intensively (see Friedrich \cite{HFri2004}, Christodoulou-Klainerman \cite{ChriKla}, Corvino \cite{Co2000}, Chrusciel and Delay \cite{ChruDe2002,ChruDe2003}, Corvino-Schoen \cite{CoScho2003} and Klainerman-Nicol{\`o} \cite{KlaNi,KlaNi2002,KlaNi2003}) and is now fairly well understood, at least in the flat case. However, it is not yet clear, given an asymptotically flat spacetime, which class of initial data yield solutions that admit the peeling property at a given order, and whether these classes are smaller than in Minkowski spacetime or not.

The works of L. Mason and J.-P. Nicolas in 2009 and 2012 \cite{MaNi2012,MaNi2009} were precisely aimed at answering this last question for the Schwarzschild metric, for scalar, Dirac and Maxwell fields. Their method combines the Penrose compactification of the spacetime and geometric energy estimates. By working in a neighbourhood of spacelike infinity on the compactified spacetime, one obtains energy estimates at all orders for the rescaled field, which control weighted Sobolev norms on $\scri$ in terms of similar norms on a Cauchy hypersurface and vice versa. The finiteness of the norms up to order $k$ at $\scri$ defines the peeling of order $k$. By completion of smooth compactly supported data on the Cauchy hypersurface in the corresponding norms, one obtains the optimal classes of data ensuring that the associated solution peels at order $k$. The result does not strictly refer to the regularity near spacelike infinity. Indeed, if the regularity is controlled in a neighbourhood of $i^0$ and on the full initial data hypersurface, it can be extended to the whole of $\scri$ by standard results for hyperbolic equations.

In this paper and the recent work \cite{NiXu2018}, we study the peeling on Kerr spacetime. The previous work \cite{NiXu2018} was an extension of the results of \cite{MaNi2009} to the Kerr metric and also treated the case of a semi-linear conformal wave equation. Here, we extend the results of \cite{MaNi2012} for Dirac fields to Kerr metrics. The method is adapted from \cite{MaNi2012,MaNi2009,NiXu2018}. On the compactified Kerr spacetime, we choose a neighbourhood of spacelike infinity that is globally hyperbolic in order to ensure that the massless Dirac equation (the Weyl equation) has a well-posed Cauchy problem. This neighbourhood is bounded in the past by the Cauchy hypersurface $\{ t=0 \}$ and in the future by future null infinity and by another null hypersurface spanned by outgoing simple null geodesics. We then obtain energy estimates both ways at all orders between the future and past boundaries of our domain, for smooth Weyl fields. An important difference with the case of the wave equation is that we have a conserved causal current, which yields a positive definite conserved quantity on spacelike hypersurfaces. This current comes from the symplectic structure associated with the equation and not from a stress-energy tensor. We therefore do not need to choose a timelike vector field with which to measure an energy. The basic energy estimate is directly an equality between $L^2$ norms on the future and past boundaries. Then, we obtain higher order estimates by commuting into the equation covariant derivatives along a family of five vector fields generating the tangent bundle. Unlike Schwarzschild, Kerr spacetime is neither spherically symmetric nor static; the consequence is that we need to control all directional derivatives at the same time, we cannot, as in the Schwarzschild case, isolate the directional derivative along outgoing principal null directions and control it on its own.

Our results are valid for all values of the mass and angular momentum, including for extreme and fast Kerr metrics. In the slow and extreme cases, we have a well-posed Cauchy problem in the black hole exterior and our results in the neighbourhood of spacelike infinity describe the asymptotic behaviour of solutions living in the whole exterior of the black hole. In the fast case, we have a naked singularity and a time machine and the Cauchy problem may not be well-posed, hence are results only describe the asymptotic behaviour of solutions in the whole spacetime if there are such solutions in the energy space considered.

The paper is organised as follows: in Section 2, we recall the definition of the Kerr metric, construct its conformal compactification, define the neighbourhood $\Omega_{^*t_0}^+$ of spacelike infinity on which we shall perform our estimates as well as its foliation by spacelike slices ${\cal H}_s$ (which are described in detail in $\cite{NiXu2018}$), present our choice of Newman-Penrose tetrad and calculate the associated spin coefficients in both the original spacetime and the rescaled one. Section 3 gives the detailed expression of the Weyl equation, some discussions on the Cauchy problem in the neighbourhood $\Omega_{^*t_0}^+$; we also introduce the conserved current, give the associated energy equality between the future and past boundaries of $\Omega_{^*t_0}^+$ and then we obtain simplified equivalent expressions of the energies on the slices of the foliation $\{ {\cal H}_s\}_s$. The peeling itself is addressed in Section 4: we obtain the approximate conservation laws for the successive directional covariant derivatives of the field and infer energy estimates both ways at all orders between the future and past boundaries of $\Omega_{^*t_0}^+$. The finiteness of the energy norm of order $k$ on $\scri^+$ defines the peeling of order $k$. Since all the energy estimates are lossless and go both ways, the optimal classes of initial data for which the corresponding solutions satisfy a peeling of order $k$ are given by the completion of smooth compactly supported initial data in the energy norm of order $k$ on the Cauchy hypersurface. Similarly to the Schwarzschild case, the estimates are uniform in the mass and angular momentum of the spacetime on any compact domain of the form $[0,M] \times [-a,a]$. This guarantees the uniformity of the optimal classes of data in such intervals and in particular entails that the optimal classes in Kerr and in Minkowski are the same, in the sense that they are characterised by the same local regularity and fall-off at spacelike infinity. The appendix contains the proof of a technical lemma appearing in Section 4.

\subsection*{Notations and conventions}
\begin{enumerate}
\item Throughout the paper, we use the formalisms of abstract indices, 2-component spinors, Newman-Penrose and Geroch-Held-Penrose.
\item We follows the convention by Penrose and Rindler \cite{PeRi84} about the Hodge dual of a 1-form $\alpha$ on a spacetime $({\cal M},g)$ (i.e. a $4-$dimensional Lorentzian manifold that is oriented and time-oriented)
\begin{equation}\label{Hodgedual}
(*\alpha)_{abcd} = e_{abcd}{\alpha}^d,
\end{equation}
where $e_{abcd}$ is the volume form on $({\cal M},g)$, denoted simply as $\mathrm{dVol}$. We shall use the following differential operator of the Hodge star
\begin{equation*}
\d *\alpha = -\frac{1}{4}(\nabla_a\alpha^a)\mathrm{dVol}.
\end{equation*}
If $\mathcal{S}$ is the boundary of a bounded open set $\Omega$ and has outgoing
orientation, using Stokes theorem, we have
\begin{equation}\label{Stokesformula}
-4\int_{\mathcal{S}}*\alpha = \int_{\Omega}(\nabla_a\alpha^a)\mathrm{dVol}.
\end{equation}
\item Let $f(x)$ and $g(x)$ be two real functions. We write $f \lesssim g$ if there exists a constant $C \in (0,+\infty)$ such that $f(x)\leq C g(x)$ for all $x$, and write $f\simeq g$ if both $f\lesssim g$ and $g\lesssim f$ are valid.

\end{enumerate}
{\bf Acknowledgements.} 
The author is grateful to Professor Jean-philippe Nicolas for suggesting this problem and for his
constant interest in his work.

\section{Geometrical and analytical setting} 
The geometry of Kerr spacetime is fully presented in the book of O'Neill \cite{ONei95}. The compactification of the exterior domain of Kerr black hole is constructed in the papers of Nicolas \cite{Ni02} and H\"afner \cite{Ha09}. In this section we just recall some basic formulas and properties which are required in the next ones. For more details on the geometric and analytic settings for the peeling on Kerr spacetime, we refer to \cite{NiXu2018}.
\subsection{The Kerr spacetime}
In Boyer-Lindquist coordinates, Kerr spacetime is a manifold $({\cal M}={\mathbb R}_t\times {\mathbb R}_r \times S_\omega^2, \, g)$ whose the metric $g$ takes the form
\begin{equation}\label{OriginalKerrMetric}
 g=\left(1-\frac{2Mr}{\rho^2}\right)\d t^2 +\frac{4aMr\sin^2\theta}{\rho^2}\d t\d \varphi -\frac{\rho^2}{\Delta}\d r^2-\rho^2\d \theta^2-\frac{\sigma^2}{\rho^2}\sin^2\theta\d \varphi^2, 
\end{equation}
$$ \rho^2=r^2+a^2\cos^2\theta,\Delta=r^2-2Mr+a^2, $$
$$ \sigma^2=(r^2+a^2)\rho^2+2Mra^2\sin^2\theta=(r^2+a^2)^2-\Delta a^2\sin^2\theta,$$
where $M>0$ is the mass of the black hole and $a\neq 0$ is its angular momentum per unit mass. Kerr spacetime is asymptotically plat and there are only two basic Killing vectors $\partial_t$ and $\partial_{\varphi}$. 

The Kerr metric given by \eqref{OriginalKerrMetric} has two types of singularities : the curvature singularity $\rho =0$ and the coordinate singularities corresponding to event horizons $\Delta =0$. There are three types of Kerr spacetimes depending the number of horizons i.e the solutions of the equation $\Delta = 0$. 
\begin{itemize}
\item[$\bullet$] Slow Kerr spacetime for $0<|a|<M$, $\Delta$ has two roots
$$r_{\pm} = M \pm \sqrt{M^2-a^2}$$
\item[$\bullet$] Extreme Kerr spacetime for $|a|=M$, $\Delta$ has an unique root $M$.
\item[$\bullet$] Fast Kerr spacetime for $|a|>M$, $\Delta$ has no root.
\end{itemize}
The black hole only exists in slow and extreme Kerr spacetimes. Slow Kerr spacetime has three blocks: $\mathcal{B}_I = \left\{ r>r_+\right\},$ the exterior of the black hole; $\mathcal{B}_{II} = \left\{r_-<r<r_+\right\},$ the dynamic region and $\mathcal{B}_{III}=\left\{r<r_-\right\},$ the region containing the singularity and a time machine. Extreme Kerr spacetime contains only two blocks $\mathcal{B}_I$ and $\mathcal{B}_{III}$. 

In Kerr spacetime, the Weyl spinor has two double principle null directions (see \cite{ONei95}):
\begin{equation}
V^{\pm}=\frac{r^2+a^2}{\Delta}\partial_t\pm\partial_{r} + \frac{a}{\Delta}\partial_{\varphi}.
\end{equation}
By the Goldberg-Sachs theorem (see \cite{GoSa62}) the integral curves of $V^+$ (resp. $V^-$) define geodesic shear-free congruences called the outgoing (resp. incoming) principal null geodesics. The principal null geodesics have  a non trivial twist, hence they are not surface forming. Another family of null geodesics of Kerr spacetime is the simple null geodesics (see \cite{FleLu03, Ha09}) that also form incoming and outgoing congruences, but non twisting i.e they are surface forming. 

For slow and extreme Kerr spacetime ($M\geq a$), block I is globally hyperbolic (see \cite{HaNi04}) and thus admits a spin structure. Fast Kerr spaccetime has no event horizon. The singularity is naked. Moreover, there is a time machine. Therefore, the spacetime is not globally hyperbolic. Hence, the Cauchy problem of the hyperbolic equations does not make the sense. However, for the purpose of this paper, we shall work on a neighbourhood of spacelike infinity which is globally hyperbolic and so admits a spin structure. Following \cite{HaNi04}, we denote by ${\mathbb S}^A$ the spin bundle over $\mathcal{B}_I$ (or only over the neighbourhood of spacelike infinity in the fast case) and by ${\mathbb S}^{A'}$ the same bundle for the complex structure replaced by its opposite. The dual bundles will be denoted by ${\mathbb S}_A$ and ${\mathbb S}_{A'}$ respectively. An abstract tensor index $a$ is a combination of an unprimed spinor index $A$ and a primed spinor index $A'$ i.e $a = AA'$. The spin bundle $\mathbb {S}^A$ (resp. ${\mathbb{S}}^{A'}$) admits a canonical sympletic form called the Levi-Civita symbol, and denoted by $\varepsilon_{AB}$ (we denote by $\varepsilon_{A'B'}$ the conjugate object on $\mathbb{S}^{A'}$, i.e.  
$\varepsilon_{A'B'}:=\bar{\varepsilon}_{A'B'}=\overline{\varepsilon_{AB}}$). These sympletic structures are compatible with the metric i.e
$$g_{ab} = \varepsilon_{AB}\varepsilon_{A'B'}.$$

The determinant of $g$ is given by $\det{g} = -\rho^4\sin^2\theta$, and we give here the form of the inverse Kerr metric $g^{-1}$, which will be useful to the next sections
\begin{equation}\label{inverse}
g^{-1} = \frac{1}{\rho^2} \left( \frac{\sigma^2}{\Delta}\partial_t^2 - \frac{2aMr}{\Delta}\partial_t\partial_{\varphi} - \Delta\partial_r^2 - \partial_{\theta}^2 - \frac{\rho^2 - 2Mr}{\Delta\sin^2\theta}\partial_{\varphi}  \right).
\end{equation}

\subsection{Rescaled Kerr spacetime}
In this paper, the Penrose compactification of Kerr spacetime is constructed using the principal null geodesics $V^{\pm}$. The star-Kerr coordinates $({^*t},r,\theta,{^*\varphi})$ are defined as
$${^*t} = t - r_*, \, {^*\varphi} = \varphi - \Lambda(r),$$
where the function $r_*$ is Regge-Wheeler-type variable and $\Lambda$ sastifies
$$\frac{\d r_*}{\d r} = \frac{r^2+a^2}{\Delta}, \,  \frac{\d \Lambda}{\d r} = \frac{a}{\Delta} \, .$$
In these coordinates, the outgoing principal null geodesics are the $r$-coordinate lines and the metric has the form
$$g = \left( 1 - \frac{2Mr}{\rho^2} \right)\d {^*t}^2 + \frac{4aMr\sin^2\theta}{\rho^2}\d{^*t}\d{^*\varphi} - \rho^2\d\theta^2 - \frac{\sigma^2}{\rho^2}\sin^2\theta\d{^*\varphi} + 2\d{^*t}\d r - 2a\sin^2\theta\d{^*\varphi}\d r.$$

We now consider the coordinates $({^*t}, \, R =1/r,\, \theta,\, {^*\varphi})$, and rescale the Kerr metric with the conformal factor $\Omega^2=R^2$, as follows
\begin{eqnarray}
\hat{g}: = R^2 g&=& R^2\left(1-\frac{2Mr}{\rho^2}\right)\d {^*t}^2+\frac{4MaR\sin^2\theta}{\rho^2}\d {^*t} \d {^*\varphi} \nonumber\\
&& -(1+a^2R^2\cos^2\theta)\d {\theta^2} -\left(1+a^2R^2+\frac{2Ma^2R\sin^2\theta}{\rho^2}\right)\sin^2\theta \d {^*\varphi^2}\nonumber\\
&&-2 \d {^*t} \d R+2a\sin^2\theta \d {^*\varphi} \d R. \label{RescMet}
\end{eqnarray}
The rescaled metric extends smoothly and non degenerately to the null hypersurface $\scri^+ = \mathbb{R}_{^*t}\times \left\{ R = 0 \right\} \times S_{\omega}^2$, which can be added as a smooth boundary to Kerr spacetime and will be called future null infinity. The Levi-Civita symbols must be rescaled as 
$$\hat{\varepsilon}_{AB}=R\varepsilon_{AB}.$$
The rescaled metric has the inverse form
\begin{eqnarray}\label{InvertResMetric}
{\hat g}^{-1}&=&-\frac{1}{\rho^2}\left(r^2a^2\sin^2\theta\partial^2_{^*t}+2(r^2+a^2)\partial_{^*t}\partial_R+2ar^2\partial_{^*t}\partial_{^*\varphi}+2a\partial_R\partial_{^*\varphi} \right) \nonumber\\
&&-\frac{1}{\rho^2} \left(R^2\Delta\partial^2_R+r^2\partial_{\theta}^2+\frac{r^2}{\sin^2\theta}\partial^2_{^*\varphi}\right). \label{InRescMet}
\end{eqnarray}

We will use the normalized Newman-Penrose tetrad given by H\"afner and Nicolas in \cite{HaNi04}. More precisely
\begin{eqnarray}
l^a\partial_a &=& \frac{1}{\sqrt{2\Delta\rho^2}} \left( (r^2+a^2)\partial_t + \Delta \partial_r + a \partial_\varphi \right), \label{New0}\\
n^a \partial_a &=& \frac{1}{\sqrt{2\Delta\rho^2}} \left( (r^2+a^2)\partial_t - \Delta \partial_r + a \partial_\varphi \right), \label{New1}\\
m^a\partial_a &=& \frac{1}{p\sqrt 2}\left(  ia\sin\theta \partial_t + \partial_\theta + \frac{i}{\sin\theta} \partial_\varphi  \right), \, \mbox{where} \,\,  p = r + ia \cos\theta.\label{New2}
\end{eqnarray}
The dual tetrad of a $1-$form is given as follows
\begin{eqnarray}
l_a \d x^a &=& \sqrt{\frac{\Delta}{2\rho^2}} \left( \d t - \frac{\rho^2}{\Delta} \d r - a\sin^2\theta \d \varphi \right), \label{DualNew0}\\
n_a \d x^a &=& \sqrt{\frac{\Delta}{2\rho^2}} \left( \d t + \frac{\rho^2}{\Delta} \d r - a\sin^2\theta \d \varphi \right), \label{DualNew1}\\
m_a \d x^a &=& \frac{1}{p\sqrt 2} \left( ia\sin\theta \d t - \rho^2 \d \theta - i(r^2+a^2)\sin\theta \d \varphi \right).\label{DualNew2}
\end{eqnarray}
We define the rescaled tetrad $\hat{l}^a, \, \hat{n}^a, \, \hat{m}^a,\, \bar{\hat{m}}^a,$ which is normalized with respect to $\hat{g}$, as follows:
$$ \hat l^a = r^2 l^a, \, \hat n^a = n^a, \, \hat m^a = rm^a $$
and we have
$$ \hat{l}_a = l_a, \, \hat{n}_a = R^2n_a, \, \hat{m}_a = Rm_a.$$
In the coordinates $({^*t}, \, R, \, \theta, \, {^*\varphi})$, the rescaled Newman Penrose tetrad becomes
\begin{eqnarray}
\hat{l}^a \partial_a &=& -\sqrt\frac{\Delta}{2\rho^2}\partial_R,\label{resNew0}\\
\hat{n}^a \partial_a &=& \sqrt\frac{2}{\Delta \rho^2} \left( (r^2+a^2)\partial_{^*t}+a\partial_{^*\varphi}+\frac{R^2\Delta}{2}\partial_R \right),\label{resNew1}\\
\hat{m}^a \partial_a &=& \frac{r}{p\sqrt{2}} \left(ia\sin\theta\partial_{^*t}+ \partial_\theta + \frac{i}{\sin \theta} \partial_{^*\varphi} \right) \label{resNew2}
\end{eqnarray} 
and
\begin{eqnarray}
\hat l_a \d x^a &=& \sqrt{\frac{\Delta}{2\rho^2}} \left( \d {^*t} - a\sin^2\theta \d {^*\varphi} \right),\label{dualNew1}\\
\hat n_a \d x^a &=& \sqrt{\frac{\Delta}{2\rho^2}} \left( R^2 \d {^*t} - \frac{2\rho^2}{\Delta}\d R - aR^2\sin^2\theta \d {^*\varphi} \right),\label{dualNew2}\\
\hat m_a \d x^a &=& \frac{1}{p\sqrt 2} \left( iaR\sin\theta \d {^*t} - R\rho^2 \d \theta - iR(r^2+a^2)\sin\theta \d {^*\varphi} \right).\label{dualNew3}
\end{eqnarray}
In terms of the associated spin-frame $\left\{ o^A,\iota^A \right\}$, the above relation is equivalent to the following rescaling
$$ \hat o^A = r o^A, \, \hat \iota^A = \iota^A, \, \hat{o}_A = o_A, \, \hat{\iota}_A = R\iota_A.$$
Denoting $\hat{l}_a\d x^a,\, \hat{n}_a \d x^a, \, \hat{m}_a\d x^a$ and $\bar{\hat m}_a \d x^a$ by $\hat{l},\, \hat{n}, \, \hat{m}$ and $\bar{\hat{m}}$ respectively, the $4$-volume measure associated with the metric $\hat{g}$ can be also calculated as
$$ \mathrm{dVol} = i \hat{l} \wedge \hat{n} \wedge \hat{m} \wedge \bar{\hat{m}} = R^2\rho^2 \d {^*t} \wedge \d R \wedge \d^2 \omega.$$ 

\subsection{Rescaling of spin coefficients}
Using the Newman-Penrose tetrad normalization  \eqref{New0}, \eqref{New1} and \eqref{New2}, the spin coefficients are calculated (see \cite{HaNi04}):
\begin{equation} \label{spin0}
\kappa= \tilde{\sigma} = \lambda = \nu = 0, 
\end{equation} 
\begin{equation} \label{spin1}
\tau = -\frac{ia\sin\theta}{\sqrt 2 \rho^2} , \, \pi = \frac{ia\sin\theta}{\sqrt 2 \bar p^2}, \, 
\tilde{\rho} = \mu = -\frac{1}{\bar p}\sqrt{\frac{\Delta}{2\rho^2}}, \, \varepsilon = \frac{Mr^2 - a^2(r\sin^2\theta + M\cos^2\theta)}{2\rho^2\sqrt{2\Delta\rho^2}},
\end{equation}
\begin{equation} \label{spin2}
\alpha = \frac{1}{\sqrt{2}\bar{p}}\left( \frac{ia\sin\theta}{\bar {p}}-\frac{\cot\theta}{2} + \frac{a^2\sin\theta\cos\theta}{2\rho^2} \right), \, \beta = \frac{1}{\sqrt{2}p} \left( \frac{\cot\theta}{2} +\frac{a^2\sin\theta\cos\theta}{2\rho^2} \right),
\end{equation} 
\begin{equation} \label{spin3}
\gamma=\frac{Mr^2 - a^2(r\sin^2\theta + M\cos^2\theta)}{2\rho^2\sqrt{2\Delta \rho^2}} - \frac{ia\cos\theta}{\rho^2}\sqrt{\frac{\Delta}{2\rho^2}}, 
\end{equation} 
here we use $\tilde{\sigma}$ and $\tilde{\rho}$ to avoid the confusion with the parameters $\sigma$ and $\rho$ in the Kerr metric.

We shall use the standard Newman-Penrose notations $D, \, D', \, \delta$ and $\delta'$ for the directional derivatives $l^a\nabla_a, \, n^a\nabla_a, \, m^a\nabla_a$ and $\bar{m}^a\nabla_a$ respectively. The relation between the rescaled spin coefficients and original ones are given in the following table (see Penrose and Rindler \cite{PeRi84}, Vol$1$, page 359)
\begin{table}[htb]
\[ \begin{array}{|c|c|c|} \hline {\hat{\kappa} = \Omega^{-3} \kappa} & {\hat{\varepsilon} = \Omega^{-2} \varepsilon} & {\hat{\pi} = \Omega^{-1} (\pi + \delta'\omega) } \\ \hline {\hat{\rho} = \Omega^{-2} ( \tilde{\rho} - D\omega ) } & {\hat{\alpha} = \Omega^{-1}( \alpha - \delta'\omega ) } & {\hat{\lambda} = \lambda} \\ \hline {\hat{\sigma} = \Omega^{-2} \tilde{\sigma}} & {\hat{\beta} = \Omega^{-1} \beta} & {\hat{\mu} = \mu + D'\omega } \\ \hline {\hat{\tau} = \Omega \tau} & {\hat{\gamma} = \gamma - D'\omega } & {\hat{\nu} = \Omega \nu } \\ \hline \end{array} \]
\caption{Behaviour of spin-coefficients under rescalling.}
\end{table}

Where the derivatives along the original frame vectors of $\omega = \log \Omega = - \log r$ are
\begin{gather}
D\omega = \frac{1}{\sqrt{2\Delta\rho^2}} \left( (r^2+a^2)\partial_t + \Delta \partial_r + a \partial_\varphi \right) (- \log r) = -\frac{1}{r} \sqrt{\frac{\Delta}{2\rho^2}}, \\
\delta' \omega = \frac{r}{\bar p\sqrt{2}} \left(-ia\sin\theta\partial_{^*t}+ \partial_\theta - \frac{i}{\sin \theta} \partial_{^*\varphi} \right) (- \log r)= 0,\\
\delta \omega = \frac{r}{p\sqrt{2}} \left(ia\sin\theta\partial_{^*t}+ \partial_\theta + \frac{i}{\sin \theta} \partial_{^*\varphi} \right) (- \log r) = 0, \\
D'\omega = \frac{1}{\sqrt{2\Delta\rho^2}} \left( (r^2+a^2)\partial_t - \Delta \partial_r + a \partial_\varphi \right) (- \log r) = \frac{1}{r} \sqrt{\frac{\Delta}{2\rho^2}}.
\end{gather}
Putting these together with the values of original spin coefficients, we obtain the rescaled spin coefficients 
\begin{equation} \label{ResSpin-coffi1}
 \hat\kappa=\hat\sigma=\hat\lambda=\hat\nu=0, 
\end{equation} 
\begin{equation} \label{ResSpin-coffi2}
\hat\tau=-\frac{ia\sin\theta r}{\sqrt 2 \rho^2} , \, \hat\pi=\frac{ia\sin\theta r}{\sqrt 2 \bar p^2}, \, 
\hat\rho = -\frac{iar\cos\theta}{\bar p}\sqrt{\frac{\Delta}{2\rho^2}}, \,  \hat\mu = \left(R-\frac{1}{\bar p}\right)\sqrt{\frac{\Delta}{2\rho^2}},
\end{equation}
\begin{equation} \label{ResSpin-coffi3}
\hat\varepsilon = \frac{Mr^4 - a^2r^2(r\sin^2\theta + M\cos^2\theta)}{2\rho^2\sqrt{2\Delta\rho^2}},\, 
\hat\alpha= \frac{r}{\sqrt 2\bar{p}}\left(\frac{ia\sin\theta}{\bar p}-\frac{\cot\theta}{2}+\frac{a^2\sin\theta\cos\theta}{2\rho^2}\right),
\end{equation} 
\begin{equation} \label{ResSpin-coffi4}
 \hat\beta=\frac {r}{\sqrt 2 p}\left(\frac{\cot\theta}{2}+\frac{a^2\sin\theta\cos\theta}{2\rho^2}\right),
\end{equation} 
\begin{equation} \label{ResSpin-coffi5}
 \hat\gamma=\frac{Mr^2 - a^2(r\sin^2\theta + M\cos^2\theta)}{2\rho^2\sqrt{2\Delta \rho^2}} - \left(\frac{ia\cos\theta}{\rho^2}+R\right)\sqrt{\frac{\Delta}{2\rho^2}}.
\end{equation} 
\begin{rem}
The orders in $R$ of $\hat{\tau}, \, \hat{\pi}, \, \hat{\mu}$ and $\hat{\gamma}$ are greater than or equal to one, whereas the ones of $\hat{\rho}, \, \hat{\varepsilon}, \, \hat{\alpha}$ and $\hat{\beta}$ are all zero.
\end{rem}

\subsection{Neighbourhood of spacelike infinity}
Following \cite{NiXu2018}, we work on a neighbourhood $\Omega_{^*t_0}^+ \, ({^*t_0}\ll -1)$ of spacelike infinity $i^0$ that is sufficiently far away from the black hole and singularities. It is bounded by a part of the Cauchy hupersurface $\Sigma_0 = \left\{ t = 0 \right\}$, a part of future null infinity $\scri^+$ and a null hypersurface ${\cal S}_{^*t_0}$. Before determining ${\cal S}_{^*t_0}$, we recall
the generalized Bondi-Sachs coordinate system $(\hat{t},r,\hat{\theta},\hat{\varphi})$ which is introduced by Flechter and Lund \cite{FleLu03}. The variable $\hat{t}$ is defined by
$$\hat{t} = t - \hat{r},$$
where
$$\hat{r}:= r_* + \int_0^{r_*}\left( \sqrt{1- \frac{a^2\Delta(s)}{((r(s))^2+a^2)^2}} - 1 \right)ds + a\sin\theta,$$
$r(r_*)$ being the reciprocal function of $r \mapsto r_*$ and
$$\Delta(r_*) = (r(r_*))^2 - 2Mr(r_*) + a^2.$$
The function $\hat{\varphi}$ is defined modulo a choice of constant of integration by
$$\hat{\varphi}:= \varphi - \int \frac{2aMr}{\Delta\sqrt{(r^2+a^2)^2-\Delta a^2}}dr.$$
As for the variable $\hat{\theta}\in [0,\pi]$, it is defined in an implicit manner by the following equation
$$\frac{1+\tanh\alpha\sin\hat{\theta}}{\tanh\theta+\sin\hat{\theta}} = \sin\theta,$$
where $\alpha$ is a primitive of $a/\sqrt{(r^2+a^2)^2-\Delta a^2}$.

The simple null geodesics of Kerr spacetime are defined as the $r$-coordinate lines in the generalized Bondi–Sachs coordinate system. Using these geodesics one can perform the star-Kerr and Kerr-star spacetimes. Then the Penrose compactification of Kerr spacetime can be obtained by inverting the variable $r$ in the generalized Bondi–Sachs coordinates (in details see \cite{Ha09}). The level hypersurfaces of $\hat{t}$ are null hypersurfaces with topology $\mathbb{R}\times S^2$ foliated by outgoing simple null geodesics; they intersect the $t = 0$
slice and $\scri^+$ at $2$-spheres. Now we determining the null hypersurface
${\cal S}_{^*t_0}$ as follows
$$ \mathcal{S}_{^*t_0} = \left\{ \hat{t} = {^*t_0}, \, t\geq 0 \right\} \text{ for } {^*t_0} \ll -1. $$
The neighbourhood $\Omega_{^*t_0}^+$ can be given precisely as 
$$ \Omega^+_{^*t_0} = {\cal I}^- ( \mathcal{S}_{^*t_0} ) \cap \{ t \geq 0\} $$
in the compactification domain. We foliate $\Omega_{^*t_0}^+$ by the spacelike hypersurfaces
$${\cal H}_s \, = \, \left\{^*t = -sr_*; \, \hat{t} \leq {^*t_0} \right\}, \,\,\,\, 0 \leq s \leq 1, \text{ for a given } {^*t}_0\ll -1.$$
The hypersurfaces ${\cal H}_0$ and ${\cal H}_1$ are the parts of $\scri^+$ and $\Sigma_0$ inside $\Omega_{^*t_0}^+$ respectively, we also denote ${\cal H}_0$ by $\scri^+_{^*t_0}$. Given $0 \leq s_2<s_2 \leq 1$, we will denote by $\mathcal{S}_{^*t_0}^{s_1,s_2}$ the part of $\mathcal{S}_{^*t_0}$ between $\mathcal{H}_{s_1}$ and $\mathcal{H}_{s_2}$.

We need the following lemma (its proof is evident) to establish simpler equivalent expressions in the next sections
\begin{lem}\label{epsilonestimates}
Let $\varepsilon >0$, then for ${^*t_0}\ll -1$, $|^*t_0|$ large enough, in $\Omega^+_{^*t_0}$, we have
\begin{equation*} 
1 < \frac{r_*}{r} < 1+\varepsilon, \,  0 < -^*t R \leq \frac{r_*}{r} < 1+\varepsilon \,\,\,  \mbox{and} \,\,\, 1-\varepsilon< \frac{\Delta}{r^2+a^2}<1.
\end{equation*}
\end{lem}

With the foliation $\left\{ {\cal H}_s \right\}_{0\leq s\leq 1}$, we choose an identifying vector field $\nu$ that satisfies $\nu(s)=1$ as follows
$$ \nu = r_*^2R^2 \frac{\Delta}{r^2+a^2}\vert ^*t \vert^{-1}\partial_R. $$
The $4-$volume measure $\mathrm{dVol}$ can be decomposed  into the product of $\d s$ along the integral lines of $\nu^a$ and the $3$-volume measure
$$ \nu \hook \mathrm{dVol} |_{{\cal H}_s}= -r_*^2 R^4\rho^2 \frac{\Delta}{r^2+a^2} \frac{1}{|^*t|}\d {^*t} \d^2\omega |_{{\cal H}_s} \simeq -\frac{1}{|{^*t}|} \d {^*t} \d^2\omega |_{{\cal H}_s} $$
on each slice ${\cal H}_s$.

\section{The Dirac fields}
A Dirac field is a solution of the Dirac equation and is the direct sum of a neutrino part $\chi^{A'} \in \mathbb{S}^{A'}$ and an anti-neutrino part $\psi_A \in \mathbb{S}_A$. In the massless case, the Dirac equation reduces to the Weyl anti-neutrino equation or simply called Weyl equation
\begin{equation}\label{OriginalWeylEquation}
\nabla^{AA'}\psi_{A} = 0.
\end{equation}

We denote the components of $\psi_A$ in the spin-frame $\left\{ o^A,\iota^A \right\}$ by $\psi_0$ and $\psi_1$:
$$\psi_0 = \psi_A o^A, \, \psi_1 = \psi_A\iota^A.$$
Using the Newman-Penrose formalism, the Weyl equation can be expressed under the form (see \cite{Cha}):
\begin{align}
\begin{cases}
n^{\bf a}\partial_{\bf a}\psi_0 - m^{\bf a}\partial_{\bf a}\psi_1 + (\mu - \gamma)\psi_0 + (\tau - \beta)\psi_1 = 0,\\
l^{\bf a}\partial_{\bf a}\psi_1 - \bar{m}^{\bf a}\partial_{\bf a}\psi_0 + (\alpha - \pi)\psi_0 + (\varepsilon - \tilde{\rho})\psi_1 = 0.\label{WeylEquationNewmanPenrose}
\end{cases}
\end{align}

\subsection{The rescaled equations}
The Weyl equation \eqref{OriginalWeylEquation} is conformally invariant (see Penrose and Rindler \cite{PeRi84}, Vol1, page 366) in the sense that the spinor field $\psi_A$ is the solution of \eqref{OriginalWeylEquation} if and only if the rescaled spinor field $\hat{\psi}_A=\Omega^{-1}\psi_A = r\psi_A$ is the solution of the rescaled equation
\begin{equation}\label{ResWeylEq}
\hnabla^{AA'}\hat{\psi}_A = 0. 
\end{equation} 

Decomposing by $\psi_A$ onto the spin-frame $\left\{ o^A,\iota^A \right\}$ and $\hat{\psi}_A$ onto the rescaled spin-frame $\left\{ \hat{o}^A,\hat{\iota}^A \right\}$, we have 
$$\psi_A = \psi_1o_A - \psi_0\iota_A , \, $$ 
$$\hat{\psi}_A = r\psi_A = r\psi_1 o_A - r\psi_0\iota_A = \hat{\psi}_1o_A - \hat{\psi}_0R\iota_A.$$ 
Therefore, the new components are related to the old components as follows 
$$\hat{\psi}_0 = r^2\psi_0, \, \hat{\psi}_1 = r\psi_1.$$ 
 
The equation \eqref{WeylEquationNewmanPenrose} provides the expression of the rescaled Weyl equation in the rescaled Newman-Penrose tetrad
\begin{align}\label{ResWeylEquationNewmanPenrose}
\begin{cases}
\hat{D}'\hat{\psi}_0 - \hat{\delta}\hat{\psi}_1 + (\hat{\mu} - \hat{\gamma})\hat{\psi}_0 + (\hat{\tau} - \hat{\beta})\hat{\psi}_1 = 0,\\
\hat{D}\hat{\psi}_1 - \hat{\delta}'\hat{\psi}_0 + (\hat{\alpha} - \hat{\pi})\hat{\psi}_0 + (\hat{\varepsilon} - \hat{\rho})\hat{\psi}_1 = 0, 
\end{cases}
\end{align} 
the link being 
\begin{eqnarray}
0 = \hnabla^{AA'}\hat{\psi}_A = &&\left( \hat{D}'\hat{\psi}_0 - \hat{\delta}\hat{\psi}_1 + (\hat{\mu} - \hat{\gamma})\hat{\psi}_0 + (\hat{\tau} - \hat{\beta})\hat{\psi}_1 \right) \bar{\hat{o}}^{A'} \nonumber\\
&&+ \left( \hat{D}\hat{\psi}_1 - \hat{\delta}'\hat{\psi}_0 + (\hat{\alpha} - \hat{\pi})\hat{\psi}_0 + (\hat{\varepsilon} - \hat{\rho})\hat{\psi}_1 \right)\bar{\hat{\iota}}^{A'}.\label{SpinorResWeyl}
\end{eqnarray} 
Replacing the rescaled Newman-Penrose tetrad \eqref{resNew0}, \eqref{resNew1}, \eqref{resNew2} and the rescaled spin coefficients \eqref{ResSpin-coffi1}-\eqref{ResSpin-coffi5} into \eqref{ResWeylEquationNewmanPenrose}, the detailed expression of the rescaled Weyl equation is as follows
\begin{gather}
\sqrt{\frac{2}{\Delta\rho^2}}\left( (r^2+a^2)\partial_{^*t} + a\partial_{^*\varphi} + \frac{R^2\Delta}{2}\partial_R \right)\hat{\psi}_0 - \frac{r}{\sqrt{2}p}\left( ia\sin\theta\partial_{^*t}+\partial_{\theta}+\frac{i}{\sin\theta}\partial_{^*\varphi} \right)\hat{\psi}_1 \nonumber\\
+\left( \left( 2R - \frac{r}{\rho^2} \right)\sqrt{\frac{\Delta}{2\rho^2}} - \frac{Mr^2-a^2(r\sin^2\theta+M\cos^2\theta)}{2\rho^2\sqrt{2\Delta\rho^2}} \right)\hat{\psi}_0 \nonumber\\
-\frac{r}{\sqrt{2}p}\left( \frac{ia\sin\theta}{\bar{p}} + \frac{\cot\theta}{2} + \frac{a^2\sin\theta\cos\theta}{2\rho^2}\right)\hat{\psi}_1 = 0, \label{ResWeyl1}\\
-\sqrt{\frac{\Delta}{2\rho^2}}\partial_R\hat{\psi}_1 + \frac{r}{\sqrt{2}\bar{p}}\left( ia\sin\theta\partial_{^*t} - \partial_{\theta} + \frac{i}{\sin\theta}\partial_{^*\varphi}  \right)\hat{\psi}_0 + \frac{r}{\sqrt{2}\bar{p}} \left( -\frac{\cot\theta}{2} + \frac{a^2\sin\theta\cos\theta}{2\rho^2}\right)\hat{\psi}_0 \nonumber\\
+ \left( \frac{Mr^4 - a^2r^2(r\sin^2\theta+M\cos^2\theta)}{2\rho^2\sqrt{2\Delta\rho^2}} + \frac{iar\cos\theta}{\bar{p}}\sqrt{\frac{\Delta}{2\rho^2}} \right)\hat{\psi}_1 = 0.\label{ResWeyl2}
\end{gather}
The Cauchy problem for the rescaled Weyl equation in $\Omega_{^*t_0}^+$ can be solved by the same method as for Dirac and Maxwell equations in \cite{MaNi2004} or the wave equation in \cite{Ni02}. The existence and uniqueness of the rescaled solution in ${\cal C}^{\infty}(\Omega^+_{^*t_0};\mathbb{S}_A)$ of the Cauchy problem  
\begin{align*}\label{CauchyResWeyl}
\begin{cases}
{\hnabla}^{AA'} \hat{\psi}_{A} = 0,\\
\hat{\psi}_{A}|_{{\cal{H}}_1} \in \mathcal{C}_0^\infty({\cal{H}}_1; \mathbb{S}_A),
\end{cases}
\end{align*}
allows us to extend the rescaled solution $\hat{\psi}_A$ to the boundary $\scri_{^*t_0}^+$. 
\begin{rem} 
The rescaled Weyl equation can be simplified by using the weighted differential operators of the GHP formalism (see Geroch-Held-Penrose \cite{GHP} and Penrose-Rindler \cite{PeRi84})
\begin{equation} \label{WeylEqGHP}
\left\{ \begin{array}{l}
{ \hat\thorn' \hat{\psi}_0 - \hat{\eth}  \hat{\psi}_1 = -\hat\mu{\hat\psi}_0-\hat\tau{\hat\psi}_1, } 
\\ \\
{ \hat\thorn \hat{\psi}_1 - \hat{\eth}' \hat{\psi}_0 = \hat\rho{\hat\psi}_1+\hat\pi{\hat\psi}_0,}  \end{array} \right.
\end{equation}
where  $\hat{\eth}'\hat{\psi}_0$ has weight $\left\{ 0;1 \right\}$ and $\hat{\eth}\hat{\psi}_1$ has weight $\left\{ 0;-1\right\}:$
$$\hat\thorn' \hat\psi_0 = (\hat D' - \hat\gamma)\hat\psi_0 \, , \, \hat\eth \hat\psi_1 = (\hat\delta + \hat\beta)\hat\psi_1,$$
$$\hat\thorn \hat\psi_1 = (\hat D + \hat\varepsilon)\hat\psi_1 \, , \, \hat\eth' \hat\psi_0 = (\hat\delta' - \hat\alpha)\hat\psi_0,$$
with $\hat{D}, \, \hat{D}',\, \hat{\delta}$ and $\hat{\delta}'$ are the derivatives along the rescaled Newman-Penrose formalism $\hat{l}^a\hnabla_a, \, \hat{n}^a\hnabla_a,$ $\hat{m}^a\hnabla_a$ and $\bar{\hat{m}}^a\hnabla_a$ respectively.
\end{rem} 

\subsection{Energy fluxes}
The conserved current for the rescaled Weyl equation is
\begin{align*}
\hat{J}^a = \hat{\psi}^A\bar{\hat\psi}^{A'} &= \hat{\varepsilon}^{AB}\hat{\psi}_B\hat{\varepsilon}^{A'B'}\bar{\hat{\psi}}_{B'} =\hat{\varepsilon}^{AB}\hat{\varepsilon}^{A'B'}(\hat{\psi}_1\hat{o}_B - \hat{\psi}_0\hat{\iota}_B)(\bar{\hat{\psi}}_1\bar{\hat{o}}_{B'} - \bar{\hat{\psi}}_0\bar{\hat{\iota}}_{B'})\\
&=\hat{\varepsilon}^{AB}\hat{\varepsilon}^{A'B'} (\vert\hat{\psi}_1\vert^2 \hat{o}_B\bar{\hat{o}}_{B'} - \hat{\psi}_1\bar{\hat{\psi}}_0\hat{o}_B\bar{\hat{\iota}}_{B'} - \hat{\psi}_0\bar{\hat{\psi}}_1\hat{\iota}_B\bar{\hat{o}}_{B'} + \vert\hat{\psi}_0\vert^2\hat{\iota}_B\bar{\hat{\iota}}_{B'} ) \\
&= \vert\hat{\psi}_1\vert^2\hat{l}^a\partial_a + \vert\hat{\psi}_0\vert^2\hat{n}^a\partial_a - \hat{\psi}_1\bar{\hat\psi}_0\hat{m}^a\partial_a - \hat{\psi}_0\bar{\hat{\psi}}_1\bar{\hat{m}}^a\partial_a.
\end{align*}
Its Hodge dual is given as (see \eqref{Hodgedual}):
\begin{align}\label{CurrentConserved}
\omega: &= *\hat{J}_a\d x^a = \hat{J}^a \partial_a \hook \mathrm{dVol}\nonumber\\
&= \left( \vert\hat{\psi}_1\vert^2\hat{l}^a\partial_a + \vert\hat{\psi}_0\vert^2\hat{n}^a\partial_a - \hat{\psi}_1\bar{\hat\psi}_0\hat{m}^a\partial_a - \bar{\hat{\psi}}_1\hat{\psi}_0\bar{\hat{m}}^a\partial_a  \right)\hook \mathrm{dVol}\nonumber\\
&= \left( \vert\hat{\psi}_1\vert^2\hat{l}^a\partial_a + \vert\hat{\psi}_0\vert^2\hat{n}^a\partial_a - \hat{\psi}_1\bar{\hat\psi}_0\hat{m}^a\partial_a - \bar{\hat{\psi}}_1\hat{\psi}_0\bar{\hat{m}}^a\partial_a  \right)\hook i\hat{l}\wedge\hat{n}\wedge\hat{m}\wedge\bar{\hat{m}} \nonumber\\
&= -i\hat{l}\wedge\hat{m}\wedge\bar{\hat{m}}\vert\hat{\psi}_1\vert^2 + i\hat{n}\wedge\hat{m}\wedge\bar{\hat{m}}\vert\hat{\psi}_0\vert^2 + i\hat{l}\wedge\hat{n}\wedge\hat{m}\hat{\psi}_1\bar{\hat{\psi}}_0 - i\hat{l}\wedge\hat{n}\wedge\bar{\hat{m}}\bar{\hat\psi}_1\hat{\psi}_0.
\end{align}
The energy flux on a oriented hypersurface $\cal{S}$ is defined as
$${\cal E}_{\cal S}(\hat\psi_A) = -4\int_{\cal S}\omega.$$

The equivalent form of the energy of the Dirac field on an oriented hypersurface ${\cal H}_s$ is given in the following lemma
\begin{lem}\label{EquivalentSimplerEnergy}
For $\vert{^*t_0}\vert$ large enough, the energy fluxes of $\hat{\psi}_A$ across the hypersurfaces ${\cal H}_{s}, \,\, (0\leq s \leq 1),$ have the following simpler equivalent expressions
$${\cal E}_{{\cal H}_s}(\hat{\psi}_A) \simeq \int_{{\cal H}_s} \left( \frac{R}{\vert{^*t}\vert}\vert\hat{\psi}_0\vert^2 + \vert\hat{\psi}_1\vert^2 \right)\d{^*t}\d^2\omega.$$
Moreover the energy flux across ${\cal S}_{^*t_0}$ is non negative.
\end{lem}
\begin{proof}
The energy flux across ${\cal S}_{^*t_0}$ is non negative since ${\cal S}_{^*t_0}$ is a null hypersurface, oriented by its future-pointing null normal vector field and $\hat{J}^a$ is a causal and future-pointing vector field.

Now on the hypersurfaces ${\cal H}_s, \,\, (0<s\leq 1),$ we have 
$$\d {^*t} = \frac{s(r^2+a^2)}{\Delta R^2}\d R.$$
Therefore, the rescaled dual Newman-Penrose tetrad \eqref{dualNew1}, \eqref{dualNew2} and \eqref{dualNew3} restricted on ${\cal H}_s$ is
\begin{eqnarray*}
\hat l_a \d x^a &=& \sqrt{\frac{\Delta}{2\rho^2}} \left( \d {^*t} - a\sin^2\theta \d {^*\varphi} \right),\\
\hat n_a|_{{\cal H}_s} \d x^a &=& \sqrt{\frac{\Delta}{2\rho^2}} \left( \left( R^2 - \frac{2\rho^2R^2}{s(r^2+a^2)} \right) \d {^*t} - aR^2\sin^2\theta \d {^*\varphi} \right),\\
\hat m_a \d x^a &=& \frac{1}{p\sqrt 2} \left( iaR\sin\theta \d {^*t} - R\rho^2 \d \theta - iR(r^2+a^2)\sin\theta \d {^*\varphi} \right).
\end{eqnarray*}
These equalities yield that
\begin{eqnarray*}
\hat{l}\wedge\hat{m}\wedge\bar{\hat{m}} &=& \sqrt{\frac{\Delta}{2\rho^2}}\frac{1}{2\rho^2} (-2iR^2\rho^2(r^2+a^2)\sin\theta \d{^*t}\wedge \d\theta \wedge \d{^*\varphi} \\
&& \hspace{3cm}+ 2ia^2R^2\rho^2\sin^3\theta \d{^*t}\wedge \d\theta \wedge \d{^*\varphi} )\\
&=&\sqrt{\frac{\Delta\rho^2}{2}}iR^2\d{^*t}\d^2\omega,
\end{eqnarray*}
\begin{eqnarray*}
\hat{n}|_{{\cal H}_s}\wedge\hat{m}\wedge\bar{\hat{m}} &=& \sqrt{\frac{\Delta}{2\rho^2}}\frac{1}{2\rho^2} \left( -2iR^4\rho^2\left( 1 - \frac{2\rho^2}{s(r^2+a^2)}\right)(r^2+a^2)\sin\theta\d{^*t}\wedge \d\theta \wedge \d{^*\varphi} \right.\\
&&\hspace{3cm} \left. +2ia^2R^4\rho^2\sin^3\theta \d{^*t}\wedge\d\theta\wedge\d{^*\varphi} \right) \\
&=& -\sqrt{\frac{\Delta\rho^2}{2}}i\left( 1 - \frac{2}{s} \right)R^4\d{^*t}\d^2\omega,
\end{eqnarray*}
and
\begin{eqnarray*}
\hat{l}\wedge\hat{n}|_{{\cal H}_s}\wedge \hat{m}&=& \frac{\Delta}{2\rho^2}\frac{1}{p\sqrt{2}}\left(-aR^3\rho^2\sin^2\theta \d{^*t}\wedge\d\theta\wedge\d{^*\varphi}\right.\\
&&\hspace{3cm} \left. +aR\rho^2\sin^2\theta\left( R^2 - \frac{2R^2\rho^2}{s(r^2+a^2)} \right) \d{^*t}\wedge\d\theta\wedge\d{^*\varphi}\right)\\
&=&-\frac{a\Delta R^3\bar{p}\sin\theta}{\sqrt{2}s(r^2+a^2)}\d{^*t}\d^2\omega,
\end{eqnarray*}
similarly
\begin{eqnarray*}
\hat{l}\wedge\hat{n}|_{{\cal H}_s}\wedge \bar{\hat{m}}&=& \frac{a\Delta R^3p\sin\theta}{\sqrt{2}s(r^2+a^2)}\d{^*t}\d^2\omega.
\end{eqnarray*}
Hence, the Hodge dual \eqref{CurrentConserved} restricted on ${\cal H}_s$ is
\begin{eqnarray*}
-\omega|_{{\cal H}_s}&=& \sqrt{\frac{\Delta\rho^2}{2}}R^2\vert\hat{\psi}_1\vert^2\d{^*t}\d^2\omega + \sqrt{\frac{\Delta\rho^2}{2}}\left( \frac{2}{s} - 1 \right)R^4\vert\hat{\psi}_0\vert^2 \d{^*t}\d^2\omega\\ 
&&-\frac{ia\bar{p}\Delta R^3\sin\theta}{\sqrt{2}s(r^2+a^2)}\hat{\psi}_1\bar{\hat{\psi}}_0\d{^*t}\d^2\omega + \frac{iap\Delta R^3\sin\theta}{\sqrt{2}s(r^2+a^2)}\bar{\hat{\psi}}_1 \hat{\psi}_0\d{^*t}\d^2\omega\\
&=& \sqrt{\frac{\Delta\rho^2R^4}{2}} \left( \vert\hat{\psi}_1\vert^2 + \left( \frac{2}{s} - 1 \right)R^2\vert\hat{\psi}_0\vert^2 \right)\d{^*t}\d^2\omega\\
&&- \left( \frac{ia\bar{p}\Delta R^3\sin\theta}{\sqrt{2}s(r^2+a^2)}\hat{\psi}_1\bar{\hat{\psi}}_0 - \frac{iap\Delta R^3\sin\theta}{\sqrt{2}s(r^2+a^2)}\bar{\hat{\psi}}_1\hat{\psi}_0 \right)\d{^*t}\d^2\omega.
\end{eqnarray*}
By the inequality that
\begin{eqnarray*}
\left| \frac{ia\bar{p}\Delta R^3\sin\theta}{\sqrt{2}s(r^2+a^2)}\hat{\psi}_1\bar{\hat{\psi}}_0 - \frac{iap\Delta R^3\sin\theta}{\sqrt{2}s(r^2+a^2)}\bar{\hat{\psi}}_1\hat{\psi}_0 \right| &\leq& \frac{\sqrt{2}a|p|\Delta R^3}{s(r^2+a^2)}\vert\hat{\psi}_1\vert\vert\hat{\psi}_0\vert \\
&\leq & \sqrt{\frac{\Delta\rho^2 R^4}{2}}\left( \frac{2a^2\Delta}{s(r^2+a^2)^2}\vert\hat{\psi}_1\vert^2 + \frac{R^2}{2s}\vert\hat{\psi}_0\vert^2  \right),
\end{eqnarray*}
we can estimate
\begin{equation}\label{EquivalentEnergy1}
4\int_{{\cal H}_s}\sqrt{\frac{\Delta\rho^2R^4}{2}} \left( \left( 1 - \frac{2a^2\Delta}{s(r^2+a^2)^2} \right)\vert\hat{\psi}_1\vert^2 + R^2\left( \frac{3}{2s} -1 \right)\vert\hat{\psi}_0\vert^2 \right)\d{^*t}\d^2\omega \leq {\cal E}_{{\cal H}_s}(\hat{\psi}_A)
\end{equation}
and
\begin{equation}\label{EquivalentEnergy2}
{\cal E}_{{\cal H}_s}(\hat{\psi}_A) \leq 4\int_{{\cal H}_s}\sqrt{\frac{\Delta\rho^2R^4}{2}} \left( \left( 1 + \frac{2a^2\Delta}{s(r^2+a^2)^2} \right)\vert\hat{\psi}_1\vert^2 + R^2\left( \frac{5}{2s} -1 \right)\vert\hat{\psi}_0\vert^2 \right)\d{^*t}\d^2\omega.
\end{equation}
For $\vert{^*t_0}\vert$ large enough, using Lemma \ref{epsilonestimates}: 
\begin{equation}\label{inequality1}
\frac{1-\varepsilon}{s(r^2+a^2)} < \frac{\Delta}{s(r^2+a^2)^2}<\frac{1}{s(r^2+a^2)}.
\end{equation}
Moreover
\begin{equation}\label{inequality2}
\frac{1}{2s}\leq \frac{3}{2s}-1, \, \frac{5}{2s}-1<\frac{5}{2s}, \, \frac{1}{s} = \frac{r_*}{\vert{^*t}\vert}, \, 0< \frac{r_*}{r} < 1+\varepsilon \,\,\, \hbox{for} \,\,\, {^*t_0}\ll -1.
\end{equation}
Combining the inequalities \eqref{EquivalentEnergy1}-\eqref{inequality2}, we obtain simpler equivalent expressions of the energy fluxes on ${\cal H}_{s}, \, (0<s\leq 1),$ as follows
$${\cal E}_{{\cal H}_s}(\hat{\psi}_A) \simeq \int_{{\cal H}_s} \left( \frac{R}{\vert{^*t}\vert}\vert\hat{\psi}_0\vert^2 + \vert\hat{\psi}_1\vert^2  \right)\d{^*t}\d^2\omega.$$
Now taking $s\rightarrow 0$, whence $r\rightarrow +\infty$, and using the fact that $R^2/s\simeq R/\vert{^*t}\vert \rightarrow 0$, we obtain 
$${\cal E}_{{\cal H}_0}(\hat{\psi}_A) \simeq \int_{{\cal H}_s}\vert\hat{\psi}_1\vert^2\d{^*t}\d^2\omega.$$
\end{proof}

\section{Covariant derivative approach to the Peeling}
On Schwarzschild's spacetime the peeling for Dirac fields can be obtained using two approaches: the partial or covariant derivative, in detail see \cite{MaNi2012}. The significant difference between the two approaches is that the covariant derivative tranverse to $\scri^+$ can be controlled independently of the other covariant derivatives. On Kerr spacetime, this difference does not appear, since the derivative tranverse to $\scri^+$ must be controlled by the derivatives of all directions in both approaches. In this section, we will consider the peeling of Dirac fields by extending the covariant derivative approach in the work on Schwarzschild spacetime of Mason and Nicolas (see Appendix in \cite{MaNi2012}). The covariant derivative will act on full Dirac fields. To obtain the peeling at higher orders, we will use the set of five vector fields ${\cal B} = \left\{ X_i, \, i=0,1...4 \right\}$:
$$ X_0 = \partial_{^*t}, \, X_1 = \partial_{^*\varphi}, \, X_2 = \sin{^*\varphi}\partial_\theta + \cot\theta\cos{^*\varphi}\partial_{^*\varphi} \, ,$$
$$ X_3 = \cos{^*\varphi}\partial_\theta - \cot\theta\sin{^*\varphi}\partial_{^*\varphi}, \, X_4 = \partial_R \, ,$$
where $X_1, X_2$ and $X_3$ are tangent vectors to $2-$sphere $S^2_{\theta,{^*\varphi}}$.

\subsection{Curvature spinors}
On a spacetime $({\cal M},g)$ admitting a spin structure and equipped with the Levi-Civitta connection, the Riemann tensor $R_{abcd}$ can be decomposed as follows (see eq.(4.6.1) page 231 in Penrose and Rindler \cite{PeRi84}, Vol$1$): 
\begin{equation*}\label{decomposition_Riemann}
 R_{abcd} = X_{ABCD} \, \varepsilon_{A'B'} \varepsilon_{C'D'} + \Phi_{ABC'D'} \, \varepsilon_{A'B'} \varepsilon_{CD} + \bar{\Phi}_{A'B'CD} \, \varepsilon_{AB} \varepsilon_{C'D'} + \bar{X}_{A'B'C'D'} \, \varepsilon_{AB} \varepsilon_{CD},
\end{equation*} 
where $X_{ABCD}$ is a complete contraction of the Riemann tensor in its primed spinor indices
$$ X_{ABCD} = \frac{1}{4} R_{abcd} {\varepsilon}^{A'B'} {\varepsilon}^{C'D'},$$
and ${\Phi}_{ab} = {\Phi}_{(ab)}$ is the trace-free part of the Ricci tensor multiplied by $-1/2$:
\[ 2{\Phi}_{ab} = 6 {\Lambda} {g}_{ab} - {R}_{ab},~ {\Lambda} = \frac{1}{24} \mathrm{Scal}_{g}. \]
We set
\[ P_{ab} = \Phi_{ab} - \Lambda g_{ab}, \]
and
\[ X_{ABCD} = \Psi_{ABCD} + \Lambda \left( \varepsilon_{AC} \varepsilon_{BD} + \varepsilon_{AD} \varepsilon_{BC} \right),~ \Psi_{ABCD} = X_{(ABCD)} = X_{A(BCD)}. \]
Under a conformal rescaling $\hat g=\Omega^2 g$, we have (see Penrose and Rindler \cite{PeRi84}, Vol$2$, pages 120-123):
\begin{gather*}
\hat{\Psi}_{ABCD} = \Psi_{ABCD},\\
\hat{\Lambda} = \Omega^{-2} \Lambda + \frac{1}{4} \Omega^{-3} \square \Omega, ~\square = \nabla^a \nabla_a, \\
\hat{P}_{ab} = P_{ab} - \nabla_b \Upsilon_a  + \Upsilon_{AB'} \Upsilon_{BA'},~ \mbox{with } \Upsilon_a = \Omega^{-1} \nabla_a \Omega = \nabla_a \log \Omega.
\end{gather*}

\begin{lem}\label{shortly.Psi}
The scalar curvature of the rescaled star-Kerr metric is
$$ \hat\Lambda=\frac{Mr-a^2}{2\rho^2}.$$
And the simpler expression of the rescaled trace-free Ricci tensor is
\begin{eqnarray*}
\hat{\Phi}_{ab}\d x^a\d x^b &=& A_1 \d{^*t}^2 + A_2 \d{^*t}\d R + A_3 \d{^*t}\d\theta + A_4\sin^2\theta\d{^*t}\d{^*\varphi} + A_5 \d R^2 + A_6\d R \d \theta \\
&& + A_7\sin^2\theta \d R \d {^*\varphi} + A_8\d \theta^2 + A_9\sin^2\theta\d \theta\d {^*\varphi} + A_{10}\sin^2\theta\d{^*\varphi}^2, 
\end{eqnarray*}
where $A_i(r,\theta), \, (i=1,2...10),$ are bounded functions in $\Omega_{^*t_0}^+$. 
\end{lem}

\begin{proof}
Since the Kerr metric is Ricci flat, $\Phi_{ab}=P_{ab}=R_{ab}=\Lambda=0$. Recall from Equation \eqref{inverse} (Section 2.1), the inversed Kerr metric is
$$g^{-1} = \frac{1}{\rho^2} \left( \frac{\sigma^2}{\Delta}\partial_t^2 - \frac{2aMr}{\Delta}\partial_t\partial_{\varphi} - \Delta\partial_r^2 - \partial_{\theta}^2 - \frac{\rho^2 - 2Mr}{\Delta\sin^2\theta}\partial_{\varphi}  \right)$$
and $\det{g} = -\rho^4\sin^2\theta$. Therefore, we have
$$ \hat\Lambda= R^{-2}\Lambda + \frac{1}{4}R^{-3}\Box_g R=\frac{1}{4}\frac{r^3}{\sqrt{\vert\det{g}}\vert}\partial_r\left(\sqrt{\vert\det{g}\vert}g^{rr}\partial_r\frac{1}{r}\right) =\frac{Mr-a^2}{2\rho^2}.$$

Thus, the rescaled trace-free of Ricci tensor is
\begin{equation}\label{TraceRicciTensor}
\hat\Phi_{ab}=\hat P_{ab}+\hat\Lambda \hat g_{ab}= \Upsilon_{AB'} \Upsilon_{BA'} - \nabla_b \Upsilon_a  + \frac{Mr-a^2}{2\rho^2} \hat{g}.
\end{equation}
We need to calculate explicitly the first two terms in the above equation. Since $\Omega = R = 1/r$,
\[ \Upsilon_a \d x^a = -\frac{\d r}{r}.\]
We first determine the spinor components. Denoting $x^0 = t$, $x^1 =r$, $x^2 = \theta$ and $x^3 = \varphi$, we get
$$ \Upsilon_{AA'} = \frac{-1}{r} g^1_{AA'} = \frac{-1}{r} g^{11} \varepsilon_{AB} \varepsilon_{A'B'} g_1^{BB'} = \frac{\Delta}{r\rho^2} \varepsilon_{AB} \varepsilon_{A'B'} g_1^{BB'}. $$
Using the dual Newman-Penrose tetrad \eqref{DualNew0}, \eqref{DualNew1} and \eqref{DualNew2}, we obtain  
$$ g_1^\mathbf{BB'} = \left( \begin{array}{cc} {n_1} & {-\bar{m}_1} \\ {-m_1} & {l_1} \end{array} \right) = \sqrt{\frac{\rho^2}{2\Delta}} \left( \begin{array}{cc} 1&0 \\ 0&-1 \end{array} \right), $$
and thus
$$ \Upsilon_\mathbf{AA'} = - \frac{1}{r} \sqrt{\frac{\Delta}{2\rho^2}} \left( \begin{array}{cc} 1&0\\0&-1 \end{array} \right). $$
The non zero components of $\alpha_{ab}:=\Upsilon_{AB'}\Upsilon_{BA'}$ are 
$$ \alpha_{00'00'}=\Upsilon_{00'}\Upsilon_{00'}=\alpha_{11'11'}=\Upsilon_{11'}\Upsilon_{11'}=\frac{\Delta}{2r^2\rho^2},$$
$$ \alpha_{01'10'}=\Upsilon_{00'}\Upsilon_{11'}=\alpha_{10'01'}=\Upsilon_{11'}\Upsilon_{00'}=  -\frac{\Delta}{2r^2\rho^2}.$$
Therefore,
\begin{align}\label{cofficient-dR0}
\Upsilon_{AB'}\Upsilon_{BA'}\d x^a\d x^b&= \frac{\Delta}{2r^2\rho^2}  (l_al_b+n_an_b-m_a{\bar m}_b-{\bar m}_a m_b) \d x^a \d x^b \nonumber\\
&=  \frac{\Delta}{2\rho^2} (R^2{\hat l}_a{\hat l}_b+R^{-2}\hat n_a\hat n_b-\hat m_a{\bar {\hat m}}_b-{\bar {\hat m}}_a \hat m_b) \d x^a \d x^b \nonumber\\
&= \frac{\Delta}{2\rho^2} (R^2\hat{l}_a\hat{l}_b + R^{-2}\hat{n}_a\hat{n}_b - 2\Re(\hat{m}_a\bar{\hat{m}}_b)) \d x^a\d x^b.
\end{align}
Using the rescaled dual Newman-Penrose tetrad \eqref{dualNew1}, \eqref{dualNew2} and \eqref{dualNew3}, we obtain that
\begin{eqnarray}\label{EquationRicci1}
\Upsilon_{AB'}\Upsilon_{BA'} &=& \frac{\d R^2}{R^2} + \frac{1}{2}\left( \frac{\Delta}{\rho^2} \right)^2 \left[ \left(1 - \frac{a^2\sin^2\theta}{\Delta} \right)R^2\d{^*t}^2 + \frac{4aMR\sin^2\theta}{\Delta}\d{^*t}\d{^*\varphi} \right. \nonumber\\
&&\hspace{3cm} \left.  - \frac{R^2\rho^4}{\Delta}\d \theta^2 - \left( \frac{(r^2+a^2)^2}{\Delta} + a^2\sin^2\theta \right)R^2\sin^2\theta \d{^*\varphi}^2  \right. \nonumber\\
&&\hspace{6,2cm}\left. - \frac{2\rho^2}{\Delta}\d{^*t}\d R + \frac{2a\rho^2}{\Delta}\sin^2\theta \d{^*\varphi}\d R \right].
\end{eqnarray}
On the other hand,
\begin{equation}\label{cofficient-dR2}
 \nabla_b \Upsilon_a \d x^a \d x^b = \nabla_b \left( -\frac{\d r}{r} \right) \d x^b = \frac{\d r^2}{r^2} + \frac{1}{r} \Gamma^1_\mathbf{ab} \d x^\mathbf{a} \d x^\mathbf{b}=\frac{\d R^2}{R^2}+\frac{1}{r} \Gamma^1_\mathbf{ab} \d x^\mathbf{a} \d x^\mathbf{b}.
\end{equation}
Among the Christoffel symbols
\[ \Gamma^1_\mathbf{ab} = \frac{1}{2} g^{1\mathbf{c}} \left( \frac{\partial g_\mathbf{ac}}{\partial x^\mathbf{b}} + \frac{\partial g_\mathbf{bc}}{\partial x^\mathbf{a}} - \frac{\partial g_\mathbf{ab}}{\partial x^\mathbf{c}} \right) =  -\frac{\Delta}{2\rho^2} \left( \frac{\partial g_{\mathbf{a}1}}{\partial x^\mathbf{b}} + \frac{\partial g_{\mathbf{b}1}}{\partial x^\mathbf{a}} - \frac{\partial g_\mathbf{ab}}{\partial r} \right), \]
the non-zero ones are
$$ \Gamma_{00}^1=\frac{\Delta}{2\rho^2}\left(1- \frac{2Mr}{\rho^2}\right)_r' = \frac{\Delta}{2\rho^2}\left( \frac{\Delta}{\rho^2} \right)'_r, \, \Gamma_{03}^1=2aM\sin^2\theta \frac{\Delta}{2\rho^2}\left(\frac{r}{\rho^2}\right)_r' = 2aM\sin^2\theta \frac{\Delta(\rho^2-2r^2)}{2\rho^6},$$
$$\Gamma_{11}^1 = \frac{\Delta}{2\rho^2}\left(\frac{\rho^2}{\Delta}\right)_r' = -\frac{\rho^2}{2\Delta}\left( \frac{\Delta}{\rho^2} \right)'_r , \, \Gamma_{12}^1=\frac{\Delta}{2\rho^2}\left(\frac{\rho^2}{\Delta}\right)_\theta' = -\frac{a^2\sin2\theta}{2\rho^2}, \, \Gamma_{22}^1= -\frac{\Delta}{2\rho^2} (\rho^2)'_r = -\frac{r\Delta}{\rho^2}, $$
$$ \Gamma_{33}^1=-\frac{\Delta}{2\rho^2}\left(\frac{\sigma^2\sin^2\theta}{\rho^2}\right)_r' = -\frac{\Delta}{2\rho^2}\sin^2\theta\left( 2r + 2Ma^2\sin^2\theta \frac{\rho^2-2r^2}{\rho^4} \right).$$
Therefore,
\begin{eqnarray}
&&\nabla_b \Upsilon_a \d x^a \d x^b \nonumber\\
&=& \frac{\d R^2}{R^2} + \frac{\Delta}{2r\rho^2} \left( \left( \frac{\Delta}{\rho^2}\right)'_r \left( \d t^2 - \frac{\rho^4}{\Delta^2}\d r^2  \right) + 2aM\sin^2\theta \frac{\rho^2-2r^2}{\rho^4} \d t\d\varphi  \right) \nonumber\\
&&- \frac{\Delta}{2r\rho^2} \left( \frac{a^2\sin 2\theta}{\Delta} \d t\d \theta + 2r\d \theta^2 + \sin^2\theta\left( 2r + 2Ma^2\sin^2\theta \frac{\rho^2-2r^2}{\rho^4} \right) \d\varphi^2 \right).\label{Upsilon}
\end{eqnarray}
Replacing $\d t, \, \d \varphi$ and $\d r$ in \eqref{Upsilon} by
$$ \d t = \d {^*t} + \frac{r^2+a^2}{\Delta}\d r \, , \, \d \varphi = \d{^*\varphi} + \frac{a}{\Delta}\d r \,\,\, \hbox{and} \,\,\, \d r = -\frac{1}{R^2}\d R,$$
we get
\begin{eqnarray}\label{EquationRicci2}
&&\nabla_b\Upsilon_a \d x^a \d x^b \nonumber\\
&=& \frac{\d R^2}{R^2} + \frac{\Delta}{2r\rho^2} \left( \frac{\Delta}{\rho^2} \right)'_r\left( \d{^*t}^2 - \frac{r^2+a^2}{R^2\Delta}\d {^*t}\d R + \frac{(2\rho^2+a^2\sin^2\theta)a^2\sin^2\theta}{R^4\Delta^2}\d R^2 \right) \nonumber\\
&&+ aM\Delta\sin^2\theta \frac{\rho^2-2r^2}{r\rho^6} \left( \d{^*t}\d{^*\varphi} - \frac{a}{R^2\Delta}\d{^*t}\d R - \frac{r^2+a^2}{R^2\Delta}\d R\d {^*\varphi} + \frac{a(r^2+a^2)}{R^4\Delta^2}\d R^2 \right) \nonumber\\
&&-\frac{a^2\sin 2\theta}{2r\rho^2} \left( \d{^*t}\d \theta - \frac{r^2+a^2}{R^2\Delta}\d R\d \theta \right) - \frac{\Delta}{\rho^2}\d\theta^2 \nonumber\\
&&-\frac{\Delta\sin^2\theta}{r\rho^2}\left( r + Ma^2\sin^2\theta\frac{\rho^2-2r^2}{\rho^4} \right) \left( \d{^*\varphi}^2 + \frac{a^2}{R^4\Delta^2}\d R^2 + \frac{2a}{R^2\Delta}\d R \d{^*\varphi} \right).
\end{eqnarray}
Plugging the equations \eqref{EquationRicci1} and \eqref{EquationRicci2} into \eqref{TraceRicciTensor} and noting that
$$\left( \frac{\Delta}{\rho^2} \right)'_r = \frac{2r(Mr-a^2) + 2a^2\cos^2\theta(r-M)}{\rho^4} \simeq 2MR^2,$$
for $\vert{^*t_0}\vert$ (therefore $r$) large enough, we can find the simpler expression of the rescaled trace-free Ricci tensor as follows
\begin{eqnarray*}
\hat{\Phi}_{ab}\d x^a\d x^b &=& A_1 \d{^*t}^2 + A_2 \d{^*t}\d R + A_3 \d{^*t}\d\theta + A_4\sin^2\theta\d{^*t}\d{^*\varphi} + A_5 \d R^2 + A_6\d R \d \theta \\
&& + A_7\sin^2\theta \d R \d {^*\varphi} + A_8\d \theta^2 + A_9\sin^2\theta\d \theta\d {^*\varphi} + A_{10}\sin^2\theta\d{^*\varphi}^2, 
\end{eqnarray*}
where $A_i(r,\theta), \,\, (i=1,2...10),$ are bounded functions in our neighbourhood $\Omega_{^*t_0}^+$. This completes the proof.
\end{proof}

\subsection{The conservation laws}
Recall that the rescaled Weyl equation $\hnabla^{AA'}\hat{\psi}_A = 0$ admits the simple conservation law 
\begin{equation}\label{conservation0}
\hnabla^a\hat{J}_a = 0.
\end{equation}

To obtain the conservation laws of the rescaled fields at higher orders, we commute the covariant derivative $\hat\nabla_{X_i}$ along a vector field $X_i \in {\cal B}$ into the rescaled Weyl equation. We then have (see \cite{MaNi2012}, Equation (41), page 25):
\begin{eqnarray}\label{commute-equation}
\hnabla^{AA'} \left( \hnabla_{X_i} \hat{\psi}_A \right) &=&\left[ \hnabla^{AA'} , \hnabla_{X_i} \right] \hat{\psi}_A \nonumber\\
&=& \left( \hnabla^a ({X_i})^b \right) \hnabla_b \hat{\psi}_A  + (X_i)_b \hat{\Phi}^{ba} \hat{\psi}_A - \hat\Lambda (X_i)^a \hat{\psi}_A \nonumber\\
&=& \left( \hnabla^a (X_i)^b \right) \hnabla_b \hat{\psi}_A  + (X_i)_b \hat{\Phi}^{ba} \hat{\psi}_A - 3\frac{Mr-a^2}{\rho^2} (X_i)^a \hat{\psi}_A.
\end{eqnarray}

Note that by straightforward calculations, for every vector field $(X_i)^b \in {\mathcal B}$ the first term $\left( \hnabla^a (X_i)^b \right) \hnabla_b \hat{\psi}_A$ can be expressed as a bounded linear combination of $(\partial_{{\bf a}})^a\hnabla_{X_i}\hat\psi_A \, (X_i \in {\mathcal B})$ for $\vert{^*t_0}\vert$ large enough. For the other terms of \eqref{commute-equation}, we observe that
\begin{align*}
 X_{i \, b} \hat{\Phi}^{ba} \hat{\psi}_A - 3\frac{Mr-a^2}{\rho^2} (X_i)^a \hat{\psi}_A &= (X_i)^b \hat g^{ac} \hat\Phi_{bc} \hat\psi_A - 3\frac{Mr-a^2}{\rho^2} (X_i)^a \hat{\psi}_A \\
&= (X_i)^{\bb}\hat g^{\ba\bc} \hat\Phi_{\bb\bc} (\partial_{\bf a})^a \hat\psi_A - 3\frac{Mr-a^2}{\rho^2} (X_i)^a \hat{\psi}_A.
\end{align*}
The simpler expression of $\hat{\Phi}_{ba}\d x^b \d x^a$ in Lemma \ref{shortly.Psi} and the formula \eqref{RescMet} (resp. \eqref{InvertResMetric}) of $\hat g$ (resp. $\hat g^{-1}$) shows that the coefficients $(X_i)^{\bb}\hat g^{\ba\bc} \hat\Phi_{\bb\bc}$ are bounded for $\vert{^*t_0}\vert$ large enough, except the one with the angular singularity, i.e $\cot\theta \hat g^{33} \hat\Phi_{33}$. However, this singularity is eliminated, since $(\partial_{^*\varphi})^a$ is of order one in $\sin\theta$ (see the matrix representation \eqref{eqcor3} of $(\partial_{^*\varphi})^a$ below). Thus, the spinor field $\hnabla^{AA'}(\hat\nabla_{X_i}{\hat\psi}_A)$ has a simpler expression as follows
\begin{eqnarray}
\hnabla^{AA'} \left( \hnabla_{X_i} \hat{\psi}_A \right) = \left[ \hnabla^{AA'} , \hnabla_{X_i} \right] \hat{\psi}_A &=& \sum_{j=0}^4T_j(r,\theta) (\partial_{\bf a})^a  \hnabla_{X_j} \hat{\psi}_A+ P_i(r,\theta)(\partial_{\bf a})^a{\hat\psi}_A \nonumber\\
&& - 3\frac{Mr-a^2}{\rho^2} ({X_i})^a \hat\psi_A,  \label{conser0}
\end{eqnarray}
where $T_j(r,\theta)$ and $\sin\theta P_i(r,\theta)$ are bounded functions. We also obtain the equation for higher order derivatives by means of the commutator expansion
\begin{eqnarray}
\hnabla^{AA'} \left( \hnabla_{X_i}^k \hat{\psi}_A \right) &=& \left[ \hnabla^{AA'} , \hnabla_{X_i}^k \right] \hat{\psi}_A \nonumber = \sum_{p=0}^{k-1} \hnabla_{X_i}^{k-p-1} \left[  \hnabla^{AA'} , \hnabla_{X_i} \right] \hnabla_{X_i}^p \hat{\psi}_A \nonumber \\
&=& \sum_{p=0}^{k-1} \hnabla_{X_i}^{k-p-1} \left( T_j(r,\theta)(\partial_{\bf a})^a \hnabla_{X_j}\hnabla_{X_i}^{p} \hat{\psi}_A \right) \nonumber \\
&& + \sum_{p=0}^{k-1} \hnabla_{X_i}^{k-p-1} \left(   P_i(r,\theta)(\partial_{\bf a})^a{\hat\psi}_A - 3\frac{Mr-a^2}{\rho^2} (X_i)^a  \hnabla_{X_i}^p\hat{\psi}_A \right). \label{commutator_high_order}
\end{eqnarray}
Now we have the values of the Infeld-van der Waerden symbols as follows
\begin{gather}
\hat{g}_0^\mathbf{AA'} = \left( \begin{array}{cc} {\hat{n}_0} & {-\bar{\hat{m}}_0} \\ {-\hat{m}_0} & {\hat{l}_0} \end{array} \right) =  \left( \begin{array}{cc} R^2\sqrt{\frac{\Delta}{2\rho^2}} & \frac{iaR\sin\theta}{\sqrt 2\bar{p}}\\ -\frac{iaR\sin\theta}{\sqrt{2}p} & -\sqrt{\frac{\Delta}{2\rho^2}} \end{array} \right), \label{eqcor0}\\
 \hat{g}_1^\mathbf{AA'} = \left( \begin{array}{cc} {\hat{n}_1} & {-\bar{\hat{m}}_1} \\ {-\hat{m}_1} & {\hat{l}_1} \end{array} \right) = - \sqrt{\frac{2\rho^2}{\Delta}}\left( \begin{array}{cc} 1 & 0\\ 0 & 0 \end{array} \right), \label{eqcor1}\\
 \hat{g}_2^\mathbf{AA'} = \left( \begin{array}{cc} {\hat{n}_2} & {-\bar{\hat{m}}_2} \\ {-\hat{m}_2} & {\hat{l}_2} \end{array} \right) = \frac{R\rho^2}{\sqrt 2}\left( \begin{array}{cc} 0& \frac{1}{\bar{p}}\\\frac{1}{p}& 0 \end{array} \right), \label{eqcor2}\\
\hat{g}_3^\mathbf{AA'} = \left( \begin{array}{cc} {\hat{n}_3} & {-\bar{\hat{m}}_3} \\ {-\hat{m}_3} & {\hat{l}_3} \end{array} \right) = \left( \begin{array}{cc} -aR^2\sin^2\theta\sqrt{\frac{\Delta}{2\rho^2}} & - \frac{iR\sin\theta(r^2+a^2)}{p\sqrt 2} \\ \frac{iR\sin\theta(r^2+a^2)}{p\sqrt 2} & a\sin^2\theta\sqrt{\frac{\Delta}{2\rho^2}} \end{array} \right). \label{eqcor3}
\end{gather}
Therefore, we can see that the components of the matrices ${\hat g}_{i}^\mathbf{AA'} \, (i=0,1,2,3)$ are bounded. Thus we can write $\hat{g}^{\bf AA'}_i$ as
 \[ \hat{g}_{i}^\mathbf{AA'}= \left( \begin{array}{cc} a_{i}(r,\theta) & b_{i}(r,\theta)\\ c_{i}(r,\theta)& d_{i}(r,\theta) \end{array} \right), \]
where the functions $a_{i},b_{i},c_{i}$ and $d_{i}$ are bounded. It follows  
\begin{equation}\label{spin-matrix}
 (\partial_{x^i})^a{\hat\psi}_A=-a_{i}(r,\theta){\hat\psi}_0 o^{A'}-b_{i}(r,\theta){\hat\psi}_0 {\iota}^{A'}+c_{i}(r,\theta){\hat\psi}_1 o^{A'}+d_{i}(r,\theta){\hat\psi}_1 {\iota}^{A'}, 
\end{equation}
where $x^i = {^*t},\, R,\, {\theta},\, {^*\varphi}$. 

Although, the vector field $(X_2)^a$ (resp. $(X_3)^a$) has the angular singularity along $\partial_{^*\varphi}$, the coefficients of $(X_2)^a\hat{\psi}_A$ (resp. $(X_3)^a\hat{\psi}_A$) are still bounded due to the fact that $(\partial_{^*\varphi})^a$ is of order one in $\sin\theta$ (see equation \eqref{eqcor3}). Therefore, we can also obtain the same equation as \eqref{spin-matrix} for $(X_2)^a\hat{\psi}_A$ (resp. $(X_3)^a\hat{\psi}_A$).

Finally, using equations \eqref{conser0} and \eqref{spin-matrix}, we obtain the conversation law of ${\hat\nabla}_{X_i}{\hat\psi}_A$ as follows
\begin{lem}\label{conser1}
For every vector field $X_i\in \mathcal{B}$ and $\vert{^*t_0}\vert$ large enough, the conservation law of the Dirac fields at first order has the simpler form
\begin{gather}
\hnabla^{AA'} \left( ( \hnabla_{X_i} \hat{\psi}_A ) (\hnabla_{X_i} \bar{\hat{\psi}}_{A'} ) \right) \nonumber\\
= 2\Re \sum_{j=0}^4\left( B_{1j} (\hnabla_{X_j} \hat\Psi )_0  \overline{(\hnabla_{X_i} \hat\Psi)_0} + B_{2j} \hat{\psi}_0 \overline{(\hnabla_{X_i} \hat{\Psi})_{0}} + B_{3j} (\hnabla_{X_j} \hat\Psi )_1 \overline{(\hnabla_{X_i}\hat\Psi)_1} \right) \nonumber\\
+ 2 \Re \sum_{j=0}^4 \left( B_{4j} \hat{\psi}_1 \overline{(\hnabla_{X_i} \hat{\Psi})_{1} } + B_{5j} \hat{\psi}_0 \overline{(\hnabla_{X_i} \hat{\Psi})_{1} } + B_{6j} \hat{\psi}_1 \overline{(\hnabla_{X_i} \hat{\Psi})_{0} }  \right) \nonumber\\
+ 2\Re \sum_{j=0}^4 \left( B_{7j}({\hat\nabla}_{X_j}{\hat\Psi})_0 \overline{({\hat\nabla}_{X_i}{\hat\Psi})_1} + B_{8j}({\hat\nabla}_{X_j}{\hat\Psi})_1\overline{({\hat\nabla}_{X_i}{\hat\Psi})_0} \right), \label{ConservationCovariant}
\end{gather}
where the coefficient functions $B_{kj}(r,\theta), \,\, (k=1,2...8),$ are bounded. Furthermore, $B_{1j}$ is of order one in $R$. 
\end{lem}
\begin{proof}
Since the coefficient functions $B_{kj},\,\, (k=1,2...8),$ are clearly bounded, we only need to prove that the coefficient $B_{1j}$ is of order one in $R$. Indeed, using the expression \eqref{spin-matrix} of $( \partial_{x^k} )^a$, we establish that for two spinor fields $\hat\psi$ and $\hat\phi$ :
$$(\partial_{x^k})^a{\hat\psi}_A \hat\phi_{A'}= - a_{k}{\hat\psi}_0 \hat\phi_0 + b_{k}{\hat\psi}_0 \hat\phi_1 + c_{k}{\hat\psi}_1 \hat\phi_0 - d_{k}{\hat\psi}_1 \hat\phi_1.$$
We therefore need to show that $a_k$ is of order one in $R$. If $k = 0$ or $k = 3$, the matrix representations \eqref{eqcor0} and \eqref{eqcor3} yield that $a_0$ and $a_3$ are of order greater than or equal to one in $R$. If $k = 2$, the matrix representation \eqref{eqcor2} yields that $a_2 = 0$, so the first term disappears. Finally, the case $k = 1$ is the trickiest one. We have
\begin{align*}
(\hnabla^a (X_4)^b)\partial_a\otimes\partial_b &= \left( \hat g^{\ba\bd}\partial_{\bd} (X_4)^{1} + \hat g^{\ba\bd} \hat \Gamma^{1}_{\bd 1} (X_4)^{1} \right)\partial_a\otimes\partial_R  = \hat g^{\ba\bd} \hat \Gamma^{1}_{\bd1} \partial_a \otimes \partial_R \\
& = \frac{1}{2}\hat g^{\ba\bd}\hat g^{1\be} \left( \frac{\partial \hat g_{\bd\be}}{\partial R} + \frac{\partial \hat g_{1\be}}{\partial{x^{\bd}}} - \frac{\partial \hat g_{\bd 1}}{\partial {x^{\be}}} \right)\partial_a \otimes \partial_R.
\end{align*}
For the cases $\ba =0$, $\ba=2$ and $\ba=3$, using again the matrix representations \eqref{eqcor0}, \eqref{eqcor2} and \eqref{eqcor3} of $(\partial_{x^{\ba}})^a$, we have the same conclusion as above. For $\ba =1$, using the formula \eqref{RescMet} of the rescaled metric $\hat g$ and the one \eqref{InvertResMetric} of $\hat{g}^{-1}$, we obtain that 
$$\hat g^{1\bd}\hat g^{1\be} \left( \frac{\partial \hat g_{\bd\be}}{\partial R} + \frac{\partial \hat g_{1\be}}{\partial{x^{\bd}}} - \frac{\partial \hat g_{\bd 1}}{\partial {x^{\be}}} \right)$$
is of order one in $R$. Therefore, the coefficients of $( \hnabla^a (X_4)^b )\hnabla_b \hat\psi_A$ are of order one in $R$. This completes the proof.
\end{proof}

\subsection{Energy estimates and peeling}
Integrating the conservation law \eqref{conservation0} of the rescaled massless Dirac fields at the zero order $\hnabla^a\hat{J}_a =0$ on $\Omega_{^*t_0}^+$, and using the equation \eqref{Stokesformula} we get
\begin{equation}
{\cal E}_{\scri^+_{^*t_0}}(\hat\psi_A) + {\cal E}_{{\cal S}_{^*t_0}}(\hat\psi_A) = {\cal E}_{{\cal H}_1}(\hat\psi_A).
\end{equation}
Now integrating the conservation law \eqref{ConservationCovariant} of the rescaled fields at the first order on the domain 
$$\Omega_{^*t_0}^{s_1,s_2}:= \Omega_{^*t_0}^+ \cap \left\{ s_1 < s < s_2 \right\} \,\,\, \hbox{with} \,\,\, 0\leq s_1<s_2 \leq 1,$$ 
and using again the equation \eqref{Stokesformula} we get
\begin{gather}
\left| {\cal E}_{{\cal H}_{s_1}} (\hnabla_{X_i}\hat{\psi}_A) + {\cal E}_{{\cal S}_{^*t_0}^{s_1,s_2}} (\hnabla_{X_i} \hat{\psi}_A) - {\cal E}_{{\cal H}_{s_2}} (\hnabla_{X_i}\hat\psi_A) \right| \nonumber\\
\lesssim \int_{s_1}^{s_2} \int_{{\cal H}_s} \left| \sum_{j=0}^4 B_{1j} (\hnabla_{X_j} \hat\Psi )_0  \overline{(\hnabla_{X_i} \hat\Psi)_0} + \sum_{j=0}^4 B_{2j} \hat{\psi}_0 \overline{(\hnabla_{X_i} \hat{\Psi})_{0} } + \sum_{j=0}^4 B_{3j} (\hnabla_{X_j} \hat\Psi )_1 \overline{(\hnabla_{X_i}\hat\Psi)_1}  \right| \frac{1}{|^*t|} \d {^*t} \d^2\omega \d s \nonumber\\
+ \int_{s_1}^{s_2} \int_{{\cal H}_s} \left| \sum_{j=0}^4B_{4j} \hat{\psi}_1 \overline{(\hnabla_{X_i} \hat{\Psi})_{1} } + \sum_{j=0}^4B_{5j} \hat{\psi}_0 \overline{(\hnabla_{X_i} \hat{\Psi})_{1} } + \sum_{j=0}^4B_{6j}\hat{\psi}_1 \overline{(\hnabla_{X_i} \hat{\Psi})_{0} } \right| \frac{1}{|^*t|} \d {^*t} \d^2\omega \d s \nonumber\\
+ \int_{s_1}^{s_2} \int_{{\cal H}_s} \left| \sum_{j=0}^4 B_{7j}({\hat\nabla}_{X_j}{\hat\Psi})_0 ({\hat\nabla}_{X_i}{\hat\Psi})_1 + \sum_{j=0}^4 B_{8j} ({\hat\nabla}_{X_j}{\hat\Psi})_1\overline{({\hat\nabla}_{X_i}{\hat\Psi})_0}  \right| \frac{1}{|^*t|} \d {^*t} \d^2\omega \d s.\label{integrating}
\end{gather}
Since $B_{1j}$ is of order one in $R$ (see Lemma \ref{conser1}), the terms containing the coefficients $B_{1j}$ can be controlled as
\begin{gather*}
\int_{s_1}^{s_2} \int_{{\cal H}_s} \left|\sum_{j=0}^4 B_{1j} (\hnabla_{X_j} \hat{\Psi})_0 \overline{(\hnabla_{X_i} \hat{\Psi})_0 } \right| \frac{1}{|^*t|} \d {^*t} \d^2\omega \d s \\
\lesssim \int_{s_1}^{s_2} \int_{{\cal H}_s} \frac{1}{\sqrt s} \left|\sum_{j=0}^4 (\hnabla_{X_j} \hat{\Psi})_0 \overline{(\hnabla_{X_i} \hat{\Psi})_0 } \right| \frac{R}{|^*t|} \d {^*t} \d^2\omega \d s\\
\lesssim \int_{s_1}^{s_2} \int_{{\cal H}_s} \frac{1}{\sqrt s}\sum_{j=0}^4\left( \frac{R}{|^*t|}\left| (\hnabla_{X_j} \hat{\Psi})_0 \right|^2 +  \frac{R}{|^*t|} \left| (\hnabla_{X_i} \hat{\Psi})_0 \right|^2 \right) \d {^*t} \d^2\omega \d s\\
\lesssim \int_{s_1}^{s_2} \frac{1}{\sqrt s}\sum_{j=0}^4 {\cal E}_{{\cal H}_s}(\hnabla_{X_j} \hat{\psi}_A) \d s.
\end{gather*}
Using $\frac{1}{|^*t|}\leq \frac{1}{\sqrt s}$, the terms containing the coefficients $B_{3j}$ and $B_{4j}$ can be controlled clearly by 
$$\int_{s_1}^{s_2} \frac{1}{\sqrt s} \left( {\cal E}_{{\cal H}_s} (\hat\psi_A) + \sum_{j=0}^4{\cal E}_{{\cal H}_s}(\hnabla_{X_j}\hat{\psi}_A) \right) \d s \, .$$ 
The terms containing the coefficients $B_{5j},\, B_{6j},\, B_{7j}$ and $B_{8j}$ can be controlled in the same manner, using the equivalence $\frac{1}{|^*t|} \simeq \frac{1}{\sqrt s} \sqrt{\frac{R}{|^*t|}},$ as $|^*t_0|$ large enough. For instance,
\begin{gather*}
\int_{s_1}^{s_2} \int_{{\cal H}_s} \left| \sum_{j=0}^4B_{5j} \hat{\psi}_0 \overline{(\hnabla_{X_i} \hat{\Psi})_1 } \right| \frac{1}{|^*t|} \d {^*t} \d^2\omega \d s \\
\lesssim \int_{s_1}^{s_2} \int_{{\cal H}_s} \left| \hat{\psi}_0 \overline{(\hnabla_{X_i} \hat{\Psi})_1 } \right| \frac{1}{\sqrt s}\sqrt{\frac{R}{|^*t|}} \d {^*t} \d^2\omega \d s\\
\lesssim \int_{s_1}^{s_2} \int_{{\cal H}_s} \frac{1}{\sqrt s}\left( \frac{R}{|^*t|}\left| \hat{\psi}_0  \right|^2 +  \left| (\hnabla_{X_i} \hat{\Psi})_1 \right|^2 \right) \d {^*t} \d^2\omega \d s\\
\lesssim \int_{s_1}^{s_2} \frac{1}{\sqrt s}\left( {\cal E}_{{\cal H}_s}(\hat\psi_A) + {\cal E}_{{\cal H}_s} (\hnabla_{X_i} \hat{\psi}_A) \right) \d s.
\end{gather*}
Lastly, the terms containing the coefficients $B_{2j}$ can be controlled by using again the equivalence $\frac{1}{|^*t|} \simeq \frac{1}{\sqrt s} \sqrt{\frac{R}{|^*t|}},$ and the following lemma (its proof will be given in the Appendix).
\begin{lem}\label{ine}
For $\vert{^*t}_0 \vert$ large enough, the energy of $\hat\psi_0$ on ${\cal H}_s$ can be controlled by the energies of the covariant derivatives $\hnabla_{X_i}\hat{\psi}_A$ and the energy of $\hat\psi_A$ on ${\cal H}_s$ uniformly in $s$. More precisely,
\begin{equation}
\int_{{\cal H}_s} \left| \hat\psi_0 \right|^2 \d^2 {^*t} \d \omega \lesssim \sum_{i=0}^4 {\cal E}_{{\cal H}_s} \left( \hnabla_{X_i}\hat\psi_A \right)  + {\cal E}_{{\cal H}_s}(\hat\psi_A).
\end{equation}
\end{lem}
Indeed
\begin{gather*}
\int_{s_1}^{s_2} \int_{{\cal H}_s} \left| \sum_{j=0}^4B_{2j} \hat{\psi}_0 \overline{(\hnabla_{X_i} \hat{\Psi})_0 } \right| \frac{1}{|^*t|} \d {^*t} \d^2\omega \d s \\
\lesssim \int_{s_1}^{s_2} \int_{{\cal H}_s} \left| \hat{\psi}_0 \overline{(\hnabla_{X_i} \hat{\Psi})_0 } \right| \frac{1}{\sqrt s}\sqrt{\frac{R}{|^*t|}} \d {^*t} \d^2\omega \d s\\
\lesssim \int_{s_1}^{s_2} \int_{{\cal H}_s} \frac{1}{\sqrt s}\left( \left| \hat{\psi}_0  \right|^2 + \frac{R}{|^*t|} \left| (\hnabla_{X_i} \hat{\Psi})_0 \right|^2 \right) \d {^*t} \d^2\omega \d s\\
\lesssim \int_{s_1}^{s_2} \frac{1}{\sqrt s}\left( \sum_{j=0}^4{\cal E}_{{\cal H}_s} ( \hnabla_{X_j}\hat\psi_A) + {\cal E}_{{\cal H}_s}(\hat\psi_A) + {\cal E}_{{\cal H}_s} (\hnabla_{X_i} \hat{\psi}_A) \right) \d s\\
\lesssim \int_{s_1}^{s_2} \frac{1}{\sqrt s}\left( \sum_{j=0}^4{\cal E}_{{\cal H}_s} ( \hnabla_{X_j}\hat\psi_A) + {\cal E}_{{\cal H}_s}(\hat\psi_A) \right) \d s.
\end{gather*}
Therefore the right-hand side of \eqref{integrating} can be controlled by 
$$\int_{s_1}^{s_2} \frac{1}{\sqrt s} \left( \sum_{j=0}^4{\cal E}_{{\cal H}_s} ( \hnabla_{X_j}\hat\psi_A) + {\cal E}_{{\cal H}_s}(\hat\psi_A) \right) \d s \, .$$
Since the function $1/\sqrt s$ is integrable on $[0,1]$, applying Gronwall's inequality we obtain the following result 
$${\cal{E}}_{\scri^+_{^*t_0}}(\hat{\psi}_A) + \sum_{X_j\in {\mathcal B}} {\cal E}_{\scri_{^*t_0}^+} \left( \hnabla_{X_j} \hat\psi_A \right) \lesssim {\cal E}_{{\cal H}_1} (\hat\psi_A) + \sum_{X_j\in \mathcal{B}} {\cal E}_{{\cal H}_1} \left( \hnabla_{X_j} \hat\psi_A \right), $$
\begin{eqnarray*}
 {\cal E}_{{\cal H}_1} (\hat\psi_A) + \sum_{X_j\in \mathcal{B}} {\cal E}_{{\cal H}_1} \left( \hnabla_{X_j} \hat\psi_A \right)  &\lesssim& \left( {\cal{E}}_{\scri^+_{^*t_0}}(\hat{\psi}_A) + {\cal{E}}_{{\cal S}_{^*t_0}}(\hat{\psi}_A)  \right) \\
&& + \sum_{X_j\in \mathcal{B}} \left\{ {\cal E}_{\scri_{^*t_0}^+} \left( \hnabla_{X_j} \hat\psi_A \right) +  {\cal E}_{{\cal S}_{^*t_0}} \left( \hnabla_{X_j} \hat\psi_A \right) \right\}.
\end{eqnarray*}

Continuing this process, we obtain the energy estimates of the Dirac fields at higher orders in the following theorem
\begin{thm}\label{peeling_covariante}
For $k \in \mathbb{N}$ and any smooth compactly supported data on ${\cal H}_1$, the associated solution $\hat\psi_A$ of the Dirac equation satisfies 
$$ \sum_{p=0}^{k}{\cal E}_{{\scri_{^*t_0}^+}} \left( \hnabla_{X_{i_j}}^p \hat\psi_A \right) \lesssim \sum_{p=0}^{k}{\cal E}_{{\cal H}_1} \left( \hnabla_{X_{i_j}}^p \hat\psi_A \right) \, , $$
$$  \sum_{p=0}^k {\cal E}_{{\cal H}_1} \left( \hnabla_{X_{i_j}}^p \hat\psi_A \right)  \lesssim
   \sum_{p=0}^k \left\{ {\cal E}_{\scri_{^*t_0}^+} \left( \hnabla_{X_{i_j}}^p \hat\psi_A \right) +  {\cal E}_{{\cal S}_{^*t_0}} \left( \hnabla_{X_{i_j}}^p \hat\psi_A \right)\right\}, $$
where
$$\hnabla_{X_{i_j}}^p \hat\psi_A = \hnabla_{X_{i_0}}\hnabla_{X_{i_1}}...\hnabla_{X_{i_p}} \hat\psi_A \, , \, X_{i_j} \in {\cal B} \, .$$  
\end{thm}
Now we define the peeling of the Dirac fields at order $k \in \mathbb{N}$ in terms of covariant derivatives as follows
\begin{defn}
We say that a solution $\psi$ of the Dirac equation peels at order $k\in \mathbb{N}$ if 
$$\sum_{p=0}^{k}{\cal E}_{\scri_{^*t_0}^+} \left( \hnabla_{X_{i_j}}^p \hat\psi_A \right) < +\infty \, .$$
\end{defn}
We obtain the optimal initial data space that guarantees the peeling at order $k$ by the following theorem
\begin{thm}
The spaces of initial data $\mathfrak {h}^k ({\cal H}_1)$ that guarantees the peeling definition at order $k\in \mathbb{N}$ is the completion of ${\cal C}_0^{\infty}(\mathcal{H}_1; \, \mathbb{S}_A)$ on the norm
$$ \left\|\hat\psi_A\right\|_{\mathfrak {h}^k ({\cal H}_1)} = \left\{\sum_{p = 0}^k{\cal E}_{{\cal H}_1} \left( \hnabla_{X_{i_j}}^p \hat\psi_A \right) \right\}^{1/2} \, .$$
\end{thm}

\appendix
\section{Proof of Lemma \ref{ine}}\label{app_2_ine2}
Using the covariant derivatives of the spin-frame along the Newman-Penrose vectors (see \cite{PeRi84}, pages 226-228), we have
\begin{gather} 
\frac{\sqrt 2 p}{r}\hat m^a + \frac{\sqrt 2 \bar p}{r}\hat {\bar m}^a = 2\partial_\theta = 2( \sin{^*\varphi} X_2 + \cos{^*\varphi} X_3) \nonumber\\
\Rightarrow \frac{\sqrt 2 p}{r}\left( \hat\nabla_{\hat m^a}\hat\Psi \right)_1 + \frac{\sqrt 2 \bar p}{r}\left( \hat\nabla_{\hat{\bar m}^a}\hat\Psi \right)_1 = 2 \left( \sin{^*\varphi} \left( \hat\nabla_{X_2}\hat\Psi \right)_1 + \cos{^*\varphi}\left( \hat\nabla_{X_3}\hat\Psi \right)_1 \right) \nonumber\\
\Leftrightarrow \frac{\sqrt 2 p}{r}\left( \hat\delta \hat\psi_1 + \hat\beta\hat\psi_1 - \hat\mu\hat\psi_0 \right) +  \frac{\sqrt 2 \bar p}{r} \left( \hat\delta'\hat\psi_1 + \hat\alpha\hat\psi_1 \right) = 2 \left( \sin{^*\varphi} \left( \hat\nabla_{X_2}\hat\Psi \right)_1 + \cos{^*\varphi}\left( \hat\nabla_{X_3}\hat\Psi \right)_1 \right) \nonumber\\
\Leftrightarrow 2\partial_\theta\hat\psi_1 = - \left( \frac{\sqrt 2 p}{r}\hat\beta + \frac{\sqrt 2 \bar p}{r}\hat\alpha \right)\hat\psi_1 + \frac{\sqrt 2 p}{r}\hat\mu\hat\psi_0 \nonumber\\
+ 2 \left( \sin{^*\varphi} \left( \hat\nabla_{X_2}\hat\Psi \right)_1 + \cos^*\varphi\left( \hat\nabla_{X_3}\hat\Psi \right)_1 \right).\label{cov-eq1}
\end{gather}
In addition, the equality
$$ \frac{\sqrt 2 p}{r}\hat m^a - \frac{\sqrt 2 \bar p}{r}\hat{\bar m}^a = 2 \left( \frac{i}{\sin\theta}\partial_{^*\varphi} + ia\sin\theta\partial_{^*t} \right)$$
yields two equations in the spin-frame $\left\{ \hat{o}^A, \, \hat{\iota}^A \right\}$:
\begin{gather}
\frac{\sqrt 2 p}{r} \left( \hat\nabla_{\hat m^a}\hat\Psi \right)_0 - \frac{\sqrt 2 \bar p}{r} \left( \hat\nabla_{\hat{\bar m}^a}\hat\Psi \right)_0 = 2\left( \frac{i}{\sin\theta}\left( \hat\nabla_{X_1}\hat\Psi \right)_0 + ia\sin\theta\left( \hat\nabla_{X_0}\hat\Psi \right)_0 \right) \nonumber\\
\Leftrightarrow \frac{\sqrt 2 p}{r}\left( -\hat\delta\hat\psi_0 - \hat\varepsilon\hat\psi_0 \right) - \frac{\sqrt 2 \bar p}{r}\left( \hat\rho\hat\psi_1 - \hat\delta'\hat\psi_0 - \hat\alpha\hat\psi_0 \right) = 2\left( \frac{i}{\sin\theta}\left( \hat\nabla_{X_1}\hat\Psi \right)_0 + ia\sin\theta\left( \hat\nabla_{X_0}\hat\Psi \right)_0 \right)\nonumber\\
\Leftrightarrow -2\left( ia\sin^2\theta\partial_{^*t} + i\partial_{^*\varphi} \right)\hat\psi_0 - \left( \frac{\sqrt 2 p}{r}\hat\varepsilon - \frac{\sqrt 2 \bar p}{r}\hat\alpha \right)\sin\theta\hat\psi_0 - \frac{\sqrt 2 \bar p}{r}\hat{\rho}\sin{\theta}\hat{\psi}_1 \nonumber\\
=  2\left( i\left( \hat\nabla_{X_1}\hat\Psi \right)_0 + ia\sin^2\theta\left( \hat\nabla_{X_0}\hat\Psi \right)_0 \right),  \label{cov-eq2}
\end{gather}
where $\left( \frac{\sqrt 2 p}{r}\hat\varepsilon - \frac{\sqrt 2 \bar p}{r}\hat\alpha \right)$ is of order zero in $R$. Moreover,
\begin{gather}
\frac{\sqrt 2 p}{r} \left( \hat\nabla_{\hat m^a}\hat\Psi \right)_1 - \frac{\sqrt 2 \bar p}{r} \left( \hat\nabla_{\hat{\bar m}^a}\hat\Psi \right)_1 = 2\left( \frac{i}{\sin\theta}\left( \hat\nabla_{X_1}\hat\Psi \right)_1 + ia\sin\theta\left( \hat\nabla_{X_0}\hat\Psi \right)_1 \right) \nonumber\\
\Leftrightarrow \frac{\sqrt 2 p}{r}\left( \hat\delta \hat\psi_1 + \hat\beta\hat\psi_1 - \hat\mu\hat\psi_0 \right) -  \frac{\sqrt 2 \bar p}{r} \left( \hat\delta'\hat\psi_1 + \hat\alpha\hat\psi_1 \right) = 2\left( \frac{i}{\sin\theta}\left( \hat\nabla_{X_1}\hat\Psi \right)_1 + ia\sin\theta\left( \hat\nabla_{X_0}\hat\Psi \right)_1 \right) \nonumber\\
\Leftrightarrow 2\left( ia\sin^2\theta\partial_{^*t} + i\partial_{^*\varphi} \right)\hat\psi_1 + \left( \frac{\sqrt 2 p}{r}\hat\beta - \frac{\sqrt 2 \bar p}{r}\hat\alpha \right)\sin\theta\hat\psi_1 - \frac{\sqrt 2 p}{r}\hat\mu \sin\theta\hat\psi_0 \nonumber\\
= 2 \left( i\left( \hat\nabla_{X_1}\hat\Psi  \right)_1 + ia\sin^2\theta\left( \hat\nabla_{X_0}\hat\Psi \right)_1 \right).\label{cov-eq3}
\end{gather}
The components of the covariant derivative $\hnabla_R \hat\psi_A$ can be calculated as follows
\begin{align}\label{cov-eq4} 
\left( \hnabla_{X_4}\hat\Psi \right)_0 &= \left( \hnabla_R\hat\Psi \right)_0 = \left( \hnabla_R\hat\Psi \right) \hat{o}^A = - \sqrt{\frac{2\rho^2}{\Delta}} \left( \hat D \hat{\psi}_A \right)\hat{o}^A \nonumber\\
&= - \sqrt{\frac{2\rho^2}{\Delta}} \left( \hat D \left( \hat{\psi}_1 \hat{o}_A - \hat{\psi}_0 \hat{\iota}_A \right) \right)\hat{o}^A = - \sqrt{\frac{2\rho^2}{\Delta}} \left( \hat{o}^A\hat D \hat{o}_A \hat{\psi}_1 - \hat{o}^A \hat{D}\hat{\psi}_0 \hat{\iota}_A \hat{\psi}_0 - \hat{D}\hat{\psi}_0 \right) \nonumber\\
&= - \sqrt{\frac{2\rho^2}{\Delta}} \left( \hat\kappa \hat\psi_1 - \hat D\hat\psi_0 - \hat\varepsilon \hat\psi_0 \right),
\end{align}
and similarly by the second equation in the GHP formalism of the Weyl system \eqref{WeylEqGHP},
\begin{align}\label{cov-eq5}
\left( \hnabla_{X_4}\hat\Psi \right)_1 &= \left( \hnabla_R\hat\Psi \right)_1 = - \sqrt{\frac{2\rho^2}{\Delta}} \left( \hat\varepsilon \hat\psi_1 + \hat D\hat\psi_1 - \hat\pi \hat\psi_0 \right) \nonumber\\
&= -\sqrt{\frac{2\rho^2}{\Delta}} \left( \hat{\eth}'\hat{\psi}_0 + \hat\rho\hat\psi_1 \right).
\end{align}

Now using the first equation of the rescaled Weyl system \eqref{ResWeyl1}, we obtain that
\begin{gather}
\left(\sin\theta - \frac{a^2\sin^3\theta}{r^2+a^2} \right)\partial_{^*t}\hat{\psi}_0 
= -\frac{\sin\theta}{r^2+a^2}\left( a(a\sin^2\theta \partial_{^*t} + \partial_{^*\varphi})\hat{\psi}_0 + \frac{R^2\Delta}{2}\partial_R\hat{\psi}_0 \right) \nonumber\\
+ \frac{1}{r^2+a^2}\sqrt{\frac{\Delta\rho^2}{2}} \left( \frac{r}{\sqrt{2}p}\left(ia\sin^2\theta\partial_{^*t}+i\partial_{^*\varphi} \right)\hat{\psi}_1 + \left( \sin\theta\partial_{\theta}+\frac{\cos\theta}{2}\right)\hat{\psi}_1 - \mathbb{V}_1(r,\theta)\hat{\psi}_0 \right) \nonumber\\
+\frac{\sin\theta}{r^2+a^2}\sqrt{\frac{\Delta\rho^2}{2}}\frac{r}{\sqrt{2}p}\left( \frac{ia\sin\theta}{\bar{p}}+\frac{a^2\sin\theta\cos\theta}{2\rho^2}\right)\hat{\psi}_1, \label{final}
\end{gather}
where $\mathbb{V}_1(r,\theta)$ is of order one in $R$. Using the equalities \eqref{cov-eq1}-\eqref{cov-eq5} we can transform the right-hand side of \eqref{final} into the sum of the components of $\hnabla_{X_i}\hat{\psi}_A$ (with $X_i \, (X_i\in{\cal B})$) and of $\hat{\psi}_A$ as 
$$\sum_{i=1}^4 \sum_{j=0}^1\alpha_{ij}(r,\theta,{^*\varphi})\left( \hnabla_{X_i}\hat\Psi \right)_j + \beta_j(r,\theta)\hat{\psi}_j,$$
where the coefficient functions $\alpha_{ij}$ and $\beta_j$ are bounded. Furthermore, $\alpha_{i0}$ and $\beta_0$ are of order greater than or equal to one in $R$. Using the fact that the norms of $\hat{\psi}_1$ and $R\hat{\psi}_0$ can be controlled by the energy of $\hat{\psi}_A$ on ${\cal H}_s$, we obtain that the norms of $\left(\sin\theta - \frac{a^2\sin^3\theta}{r^2+a^2} \right)\partial_{^*t}\hat{\psi}_0$ and then the norm of $\sin\theta\partial_{^*t}\hat{\psi}_0$ can be controlled by the energies of the covariant derivatives $\hnabla_{X_i}\hat{\psi}_A \, (X_i\in{\cal B})$ and the energy of $\hat{\psi}_A$ on ${\cal H}_s$ uniformly in $s$. On the other hand, the equation \eqref{cov-eq5} provides that
$$\left( \hnabla_{X_4}\hat\Psi \right)_1 = -\sqrt{\frac{2\rho^2}{\Delta}} \left( \frac{r}{\bar{p}}\left(\eth'-\frac{ia\sin\theta}{\sqrt 2}\partial_{^*t} \right)\hat\psi_0 + \frac{r}{\sqrt{2}\bar{p}}\left( \frac{ia\sin\theta}{\bar{p}} + \frac{a^2\sin\theta\cos\theta}{2\rho^2} \right)\hat{\psi}_0 + \hat\rho\hat\psi_1 \right),$$
where
$$\eth' = \frac{1}{\sqrt 2}\left( \partial_{\theta}-\frac{i}{\sin\theta}\partial_{^*\varphi}+\frac{\cot\theta}{2} \right).$$
Since the lowest eigenvalue of $\eth'$ on weighted scalar fields of weight $\left\{1;\, 0\right\}$ is positive (see \cite{PeRi84} section 4.15), the norm of $\hat{\psi}_0$ on $2-$sphere $S^2_{\theta,{^*\varphi}}$ can be controlled by the norm of $\eth'\hat{\psi}_0$. Therefore, we conclude that the norm of $\hat\psi_0$ on ${\cal H}_s$ can be controlled uniformly in $s$ by the ones of the covariant derivatives $\hnabla_{X_i}\hat\psi_A \, (X_i \in {\cal B})$ and $\hat{\psi}_A$ on ${\cal H}_s$.


\begin{thebibliography}{100}
\bibitem{AnBlu2015} L. Andersson and P. Blue, {\it Hidden symmetries and decay for the wave equation on the Kerr spacetime}, Annals of Mathematics, 182(3):787-853, 2015. arXiv:0908.2265.

\bibitem{AnBlu} L. Andersson and P. Blue, {\em Uniform energy bound and asymptotics for the Maxwell
field on a slowly rotating Kerr black hole exterior}, Journal of Hyperbolic Differential Equations 12.04 (2015), pp. 689-743.

\bibitem{Blu} P. Blue, {\em Decay of the maxwell field on the schwarzschild manifold}, Journal of
Hyperbolic Differential Equations 05.04 (2008), pp. 807-856.

\bibitem{Cha} S. Chandrasekhar, {\em The mathematical theory of black
        holes}, Oxford University Press 1983.

\bibitem{ChriKla} D. Christodoulou \& S. Klainerman, {\em The
global nonlinear stability of the Minkowski space}, Princeton
Mathematical Series 41, Princeton University Press 1993.

\bibitem{ChruDe2002} P. Chrusciel \& E. Delay, {\em Existence of non trivial,
    asymptotically vacuum, asymptotically simple space-times},
    Class. Quantum Grav. {\bf 19} (2002), L71-L79, erratum
    Class. Quantum Grav. {\bf 19} (2002), 3389.

\bibitem{ChruDe2003} P. Chrusciel \& E. Delay, {\em On mapping properties of
    the general relativistic constraints operator in weighted function
    spaces, with applications}, preprint Tours Univervity, 2003.

\bibitem{Co2000} J. Corvino, {\em Scalar curvature deformation and a
    gluing construction for the Einstein constraint equations},
    Comm. Math. Phys. {\bf 214} (2000), 137--189.

\bibitem{CoScho2003} J. Corvino \& R.M. Schoen, {\em On the asymptotics
    for the vacuum Einstein constraint equations}, gr-qc 0301071,
  2003.

\bibitem{DaRo} M. Dafermos \& I. Rodnianski, {\em The redshift effect and radiation decay on black hole space-times}, Comm. Pure Appl. Math. {\bf 62} (2009), 7, 859-919.

\bibitem{FiKaSmYa} F. Finster, N. Kamran, J. Smoller and S.-T. Yau, {\it a Decay of Solutions of the Wave Equation in the Kerr Geometry}, communications in mathematics physics, Vol. 264, {\bf 2}, 465-503, 2006. 

\bibitem{FleLu03} S. J. Flechter, A.W.C. Lund, {\em The Kerr spacetime in generalized Bondi-Sachs coordinates}, Class. Quantum Grav. 20 (2003), 4153–4167.

\bibitem{Fri1980} F.G. Friedlander, {\em Radiation fields and
  hyperbolic scattering theory}, Math. Proc. Camb. Phil. Soc. {\bf 88}
  (1980), 483-515.

\bibitem{Fri2001} F.G. Friedlander, {\em Notes on the wave equation on
  asymptotically Euclidean manifolds}, J. Functional Anal. {\bf 184}
  (2001), 1-18.

\bibitem{HFri2004} H. Friedrich, {\em Smoothness at null infinity and the structure of initial data}, in The Einstein equations and the large scale behavior of gravitational fields, p. 121--203, Ed. P. Chrusciel and H. Friedrich, Birkha\"user, Basel, 2004.

\bibitem{GHP} R. Geroch, A. Held \& R. Penrose, {\em A space-time calculus based on pairs of null directions}, J. Math. Phys. {\bf 14} (1973), 874-881.

\bibitem{GoSa62} J. N. Goldberg and R. K. Sachs, {\em A theorem on Petrov types}, Acta Phys. Polon., 22 (1962), 13–23, Suppl.

\bibitem{Ha09} D. Häfner, {\em Creation of fermions by rotating charged black holes}, Memoirs of the SMF 117 (2009), 158 pp, arXiv:math/0612501.

\bibitem{HaNi04} D. H\"afner \& J.-P. Nicolas, {\em Scattering of massless Dirac fields by a slow Kerr black hole}, Reviews in Mathematical Physics 16 (2004), 1, p. 29-123.

\bibitem{KlaNi} S. Klainerman \& F. Nicol{\`o}, {\em On local and
global aspects of the Cauchy problem in general relativity},
Class. Quantum Grav. {\bf 16} (1999), p. R73-R157.

\bibitem{KlaNi2002} S. Klainerman \& F. Nicol{\`o}, {\em The Evolution Problem
      in General Relativity}, Progress in Mathematical Physics Vol. 25
      (2002), Birkha\"user.

\bibitem{KlaNi2003} S. Klainerman \& F. Nicol{\`o}, {\em Peeling properties of asymptotically flat solutions to the Einstein vacuum equations}, Class. Quantum Grav. {\bf 20} (2003), p. 3215-3257.

\bibitem{Le} J. Leray, {\em Hyperbolic Differential Equations, Lecture Notes}, Princeton Institute
for Advanced Studies (1953)

\bibitem{MaNi2004} L.J. Mason, \& J.-P. Nicolas, {\em Conformal scattering and the Goursat problem}, Journal of Hyperbolic Differential Equations, {\bf 1} (2) (2004), p. 197-233.

\bibitem{MaNi2012} L.J. Mason, \& J.-P. Nicolas, {\em Peeling of Dirac and Maxwell fields on a Schwarzschild background}, J. Geom. Phys. 62 (2012), no. 4, 867-889.

\bibitem{MaNi2009} L.J. Mason, \& J.-P. Nicolas, {\em Regularity at spacelike and null infinity}, J. Inst. Math. Jussieu {\bf 8} (2009), 1, 179-208.

\bibitem{MeTaTo} J. Metcalfe, D. Tataru, and M. Tohaneanu, {\em Pointwise decay for the Maxwell field on black hole space-times}, Adv. Math., 316:53-93, 2017.

\bibitem{NePe62} E.T. Newman, R. Penrose {\em An approach to gravitational radiation by a method of spin coefficients}, J. Mathematical Phys.  {\bf 3} (1962), 566-578.

\bibitem{Ni02} J.-P.Nicolas, {\em A non linear Klein-Gordon equation on Kerr metrics}, Journal de Mathématiques Pures et Appliquées,  81 (2002), 9, p. 885-914.

\bibitem{Ni03} J.-P. Nicolas, {\em Dirac fields on asymptotically flat space-times}, Dissertationes Mathematicae 408, 2002, 85 pages.

\bibitem{Ni1997} J.-P. Nicolas, {\em Global exterior Cauchy problem for spin $3/2$ zero rest-mass fields in the Schwarzschild space-time}, Commun. in PDE 22 (1997), 3\&4, p. 465-502.

\bibitem{NiXu2018} J.-P. Nicolas, \& T.X. Pham, {\em Peeling on Kerr spacetime: linear and non linear scalar fields},  Annales Henri Poincar\'e, Vol. 20, Issue 10 (2019), p. 3419–3470. 

\bibitem{ONei95} B. O'Neill, {\em The geometry of Kerr black holes}, A.K Peters, Wellesley, 1995.

\bibitem{Pe63} R. Penrose, {\em Asymptotic properties of fields and spacetime}, Phys. Rev. Lett. {\bf 10} (1963), 66–68.

\bibitem{Pe64} R. Penrose, {\em Conformal approach to infinity}, in Relativity, groups and topology, Les Houches 1963, ed. B.S. De Witt and C.M. De Witt, Gordon and Breach, New-York, 1964.

\bibitem{Pe65} R. Penrose, {\em Zero rest-mass fields including gravitation~: asymptotic behavior}, Proc. Roy. Soc. {\bf A284} (1965), 159--203.

\bibitem{PeRi84} R. Penrose and W. Rindler, Spinors and space-time, Vol. 1 (1984) and Vol. 2 (1986), Cambridge University Press.

\bibitem{Pham2017} T.X. Pham, {\em Peeling and conformal scattering on the spacetimes of the general relativity}, Phd's thesis, Brest university (France) (4/2017), https://tel.archives-ouvertes.fr/tel-01630023/document.

\bibitem{Sa61} R. Sachs, {\em Gravitational waves in general relativity VI, the outgoing radiation condition}, Proc. Roy. Soc. London {\bf A264} (1961), 309-338.

\bibitem{Sa62} R. Sachs, {\em Gravitational waves in general relativity VIII, waves in asymptotically flat spacetime}, Proc. Roy. Soc. London {\bf A270} (1962), 103–126.

\bibitem{SmolXi} J. Smoller and Ch. Xie, {\em Asymptotic behavior of massless Dirac waves in Schwarzschild geometry}, Ann. Henri Poincar\'e 13, 943--989 (2012).
\end{thebibliography}
\end{document}